\newif\ifready\readytrue
\newif\ifcameraready\camerareadyfalse

\documentclass{article}
\usepackage[utf8]{inputenc}
\usepackage{graphicx}
\usepackage{booktabs}
\usepackage[top=1in,left=1in,right=1in,bottom=1in]{geometry}
\usepackage{amsthm}
\usepackage{amsmath}
\usepackage{amssymb}
\usepackage[colorlinks=true,linkcolor=blue,citecolor=blue,allcolors=blue,hypertexnames=false]{hyperref}
\hypersetup{
    colorlinks=true,
    linkcolor=blue,
    filecolor=blue,
    urlcolor=blue,
    allcolors=blue
}
\usepackage{relsize}
\usepackage{algorithmicx}
\usepackage{algorithm}
\usepackage[noend]{algpseudocode}
\usepackage{multicol}
\usepackage[most]{tcolorbox}
\usepackage{color}
\usepackage{pgfplots}
\usepackage{subcaption}
\pgfplotsset{
    compat=1.3,
    legend image code/.code={
        \draw [#1] (0cm,-0.1cm) rectangle (0.6cm,0.1cm);
    },
}
\usepackage{thmtools}
\usepackage{thm-restate}
\usepackage{enumerate}
\usepackage{enumitem}
\setlist{noitemsep,topsep=0pt,parsep=0pt,partopsep=0pt}
\usepackage{mathtools}
\usepackage{xspace}
\usepackage{multirow}

\usepackage[capitalize,nameinlink]{cleveref}

\theoremstyle{plain}
\newtheorem{theorem}{Theorem}[section]

\newtheorem{invariant}{Invariant}
\newtheorem{lemma}[theorem]{Lemma}

\newtheorem{definition}[theorem]{Definition}

\newtheorem{observation}[theorem]{Observation}

\newcommand{\defn}[1]{\textbf{\textit{#1}}}

\crefname{theorem}{Theorem}{Theorems}
\Crefname{lemma}{Lemma}{Lemmas}
\Crefname{claim}{Claim}{Claims}
\Crefname{observation}{Observation}{Observations}
\Crefname{algorithm}{Algorithm}{Algorithms}
\Crefname{myalgctr}{Algorithm}{Algorithms}
\Crefname{challenge}{Challenge}{Challenges}

\algrenewcommand\algorithmicindent{1em}%

\newcommand{\alg}{\mathcal{A}}

\newcommand{\mM}{\mathcal{M}}

\newcommand{\mR}{\mathcal{R}}

\newcommand{\tO}{\widetilde{O}}

\DeclarePairedDelimiter\ceil{\lceil}{\rceil}

\usepackage{xcolor}

\DeclareMathOperator{\poly}{poly}

\definecolor{mygreen}{RGB}{20,140,80}
\definecolor{linkcolor}{RGB}{0,0,230}
\definecolor{mylightgray}{RGB}{230,230,230}
\definecolor{verylightgray}{RGB}{245,245,245}

\newcommand{\etal}[0]{et al.\xspace}

\newcounter{myalgctr}

\newtcolorbox{OuterBox}[1][]{%
    breakable,
    enhanced,
    frame hidden,
    interior hidden,
    left=-5pt,
    right=-5pt,
    top=-5pt,
    float=p,
    boxsep=0pt,
    arc=0pt
#1}%

\newtcolorbox{InnerBox}[1][]{%
    enforce breakable,
    enhanced,
    colback=gray,
    colframe=white,
#1}%

\newenvironment{tbox}{
\vspace{0.2cm}
\begin{tcolorbox}[width=\columnwidth,
                  enhanced,
                  boxsep=2pt,
                  left=1pt,
                  right=1pt,
                  top=4pt,
                  boxrule=1pt,
                  arc=0pt,
                  colback=white,
                  colframe=black,
	              breakable
                  ]%
}{
\end{tcolorbox}
}

\newcommand{\tboxhrule}[0]{\vspace{0.1cm} {\color{black} \hrule} \vspace{0.2cm}}

\newenvironment{titledtbox}[1]{\begin{tbox}#1 \tboxhrule}{\end{tbox}}

\newcommand{\opt}{\textsc{OPT}}

\newcommand{\eps}{\varepsilon}

\newcommand{\whp}{\textit{whp}\xspace}

\newcommand{\geom}{\mathsf{Geom}}

\newcommand{\fact}{\log^2 n}

\newcommand{\copies}{C}
\newcommand{\mmd}{\widehat{M}_d}
\newcommand{\mwmapprox}{6\eps}
\newcommand{\bidderpot}{\Pi_{bidders}}
\newcommand{\mwm}{MWM\xspace}
\newcommand{\nonduplicate}{non-duplicate\xspace}
\newcommand{\mcbm}{\textsc{MCbM}\xspace}

\newcommand{\mcm}{\textsc{MCM}\xspace}

\newcommand{\draw}{\geom(\exp(\eps/\fact))}

\newcommand{\heps}{\eps}

\newcommand{\mwmpretransform}{\frac{\log^2\left(W\right)}{\eps^4}}
\newcommand{\phases}{\ceil{\mwmpretransform}}
\newcommand{\bmphases}{\ceil{\frac{2}{\eps^2}}}

\newcommand{\unhappy}{\textsc{Unhappy}}

\newcommand{\happy}{\textsc{Happy}}
\newcommand{\defined}{\triangleq}
\newcommand{\wopt}{W_{\opt}}
\newcommand{\hatM}{\hat{M}}

\title{Scalable Auction Algorithms for Bipartite Maximum Matching Problems}

\author{Quanquan C. Liu, Yiduo Ke, Samir Khuller}
\date{}

\begin{document}
\renewcommand{\algorithmicrequire}{\textbf{Input:}}
\renewcommand{\algorithmicensure}{\textbf{Output:}}
\algblock{ParFor}{EndParFor}
\algblock{Input}{EndInput}
\algblock{Output}{EndOutput}
\algblock{ReduceAdd}{EndReduceAdd}

\algnewcommand\algorithmicparfor{\textbf{parfor}}
\algnewcommand\algorithmicinput{\textbf{Input:}}
\algnewcommand\algorithmicoutput{\textbf{Output:}}
\algnewcommand\algorithmicreduceadd{\textbf{ReduceAdd}}
\algnewcommand\algorithmicpardo{\textbf{do}}
\algnewcommand\algorithmicendparfor{\textbf{end\ input}}
\algrenewtext{ParFor}[1]{\algorithmicparfor\ #1\ \algorithmicpardo}
\algrenewtext{Input}[1]{\algorithmicinput\ #1}
\algrenewtext{Output}[1]{\algorithmicoutput\ #1}
\algrenewtext{ReduceAdd}[2]{#1 $\leftarrow$ \algorithmicreduceadd(#2)}
\algtext*{EndInput}
\algtext*{EndOutput}
\algtext*{EndIf}
\algtext*{EndFor}
\algtext*{EndWhile}
\algtext*{EndParFor}
\algtext*{EndReduceAdd}

\maketitle
\sloppy
\thispagestyle{empty}

\pagenumbering{arabic} 

\begin{abstract}
Bipartite maximum matching and its variants are well-studied problems under various models of computation with the vast majority of approaches centering around various methods to find and eliminate augmenting paths. Beginning with the seminal papers of Demange, Gale and Sotomayor [DGS86] and Bertsekas [Ber81], bipartite maximum matching problems have also been studied in the context of \emph{auction algorithms}. These algorithms model the maximum matching problem as an auction where one side of the bipartite graph consists of bidders and the other side consists of items; as such, these algorithms offer a very different approach to solving this problem that do not use classical methods. Dobzinski, Nisan and Oren [DNO14] demonstrated the utility of such algorithms in distributed, interactive settings by providing a simple and elegant $O(\log n/\varepsilon^2)$ round \emph{maximum cardinality bipartite matching (MCM)} algorithm that has small round and communication complexity and gives a $(1-\varepsilon)$-approximation for any (not necessarily constant) $\varepsilon > 0$.  They leave as an open problem whether an auction algorithm, with similar guarantees, can be found for the \emph{maximum weighted bipartite matching (MWM)} problem. Very recently, Assadi, Liu, and Tarjan [ALT21] extended the utility of auction algorithms for MCM into the semi-streaming and massively parallel computation (MPC) models, by cleverly using maximal matching as a subroutine, to give a new auction algorithm that uses $O(1/\varepsilon^2)$ rounds and achieves the state-of-the-art bipartite MCM results in the streaming and MPC settings.

In this paper, we give new auction algorithms for maximum weighted bipartite matching (MWM) and maximum cardinality bipartite $b$-matching (MCbM). Our algorithms run in $O\left(\log n/\varepsilon^8\right)$ and $O\left(\log n/\varepsilon^2\right)$ rounds, respectively, in the blackboard distributed setting. We show that our MWM algorithm can be implemented in the distributed, interactive setting using $O(\log^2 n)$ and $O(\log n)$ bit messages, respectively, directly answering the open question posed by Demange, Gale and Sotomayor [DNO14]. Furthermore, we implement our algorithms in a variety of other models including the the semi-streaming model, the shared-memory work-depth model, and the massively parallel computation model. Our semi-streaming MWM algorithm uses $O(1/\varepsilon^8)$ passes in $O(n \log n \cdot \log(1/\varepsilon))$ space and our MCbM algorithm runs in $O(1/\varepsilon^2)$ passes using $O\left(\left(\sum_{i \in L} b_i + |R|\right)\log(1/\varepsilon)\right)$ space (where parameters $b_i$ represent the degree constraints on the $b$-matching and $L$ and $R$ represent the left and right side of the bipartite graph, respectively). Both of these algorithms improves \emph{exponentially} the dependence on $\varepsilon$ in the space complexity in the semi-streaming model against the best-known algorithms for these problems, in addition to improvements in round complexity for MCbM. Finally, our algorithms eliminate the large polylogarithmic dependence on $n$ in depth and number of rounds in the work-depth and massively parallel computation models, respectively, improving on previous results which have large polylogarithmic dependence on $n$ (and exponential dependence on $\eps$ in the MPC model). 
\end{abstract}

\section{Introduction}\label{sec:intro}

One of the most basic problems in combinatorial optimization is that of bipartite matching. This central problem has been studied extensively in 
many fields including operations 
research, economics, and computer science and is the cornerstone of many algorithm design courses and books. There is an abundance
of existing classical and recent theoretical work on this topic~\cite{Konig1916,edmonds1965maximum,Edmonds1965PathsTA,Harvey2006AlgebraicSA,HK71,LS20,madry13,MV80,ALT21,DNO14,MS04}. 
Bipartite maximum matching and its variants are commonly
taught in undergraduate algorithms courses and are so prominent to be featured regularly in competitive programming contests. In both of these settings, the 
main algorithmic solutions
for maximum cardinality matching (\mcm) 
and its closely related problems of maximum weight matching (\mwm)
are the Hungarian method using augmenting paths and reductions to maximum flow.
Although foundational, such approaches are sometimes difficult to generalize to obtain efficient solutions in other
\emph{scalable models of computation}, e.g.\ distributed, streaming, and parallel 
models.

Although somewhat less popularly known, the elegant and extremely 
simple \emph{auction-based maximum cardinality
and maximum weighted matching algorithms} of Demange, Gale, and Sotomayor~\cite{demange1986multi} and Bertsekas~\cite{bertsekas1981new} solve the 
maximum cardinality/weighted matching problems in bipartite graphs. 
Their \mcm auction algorithms denote vertices on one side of the bipartite
input as \emph{bidders} and the other side as \emph{items}. Bidders are maintained in a queue and while the queue is not empty, the first 
bidder from the queue \emph{bids} on an item with minimum price (breaking ties arbitrarily) from its neighbors. 
This bidder becomes the new owner of the item. Each time an 
item is reassigned to a new bidder, its price increases by some (not necessarily constant) $\eps > 0$. If the assigned
item still has price less than $1$, the bidder is added again to the end of the queue. Setting $\eps = \frac{1}{n + 1}$ results 
in an algorithm that gives an exact maximum cardinality matching in $O(mn)$ time, where $m$ and $n$ refer to the number of edges and vertices respectively. Such an algorithm intuitively takes advantage of the fact that 
bidders prefer items in low demand (smaller price); naturally, such items should also be matched in a maximum cardinality matching.

One of the bottlenecks in the original auction algorithm is the need to maintain bidders in a queue from which they are selected, one at a time, to bid on items. Such a bottleneck is a key roadblock to the scalability of such algorithms.
More recently, Dobzinski, Nisan, and Oren~\cite{DNO14} extended this algorithm to the approximation setting for any (not necessarily constant) $\eps > 0$. 
They give a simple and elegant randomized $(1-\eps)$-approximation algorithm for bipartite 
\mcm in $O\left(\frac{\log n}{\eps^2}\right)$ rounds of communication
for any $\eps > 0$. Furthermore, they illustrate an additional advantage for this algorithm beyond its simplicity. 
They show that in a distributed, interactive, blackboard setting, their auction 
\mcm multi-round interactive algorithm uses less communication bits than traditional algorithms for this problem. This 
interactive setting is modeled via simultaneous communication protocols where agents simultaneously send a single message in each round to a central 
coordinator and some state is computed by the central coordinator 
after each round of communication. The goal in this model is to limit the total number of bits 
sent in all of the agents' messages throughout the duration of the algorithm. They leave as an open 
question whether an interactive, approximation auction algorithm that uses approximately the same number of rounds and bits of communication
can be found for the maximum \emph{weighted} bipartite matching problem.

Such an approach led to the recent simple and elegant paper of Assadi, Liu, and Tarjan~\cite{ALT21} that adapted
their algorithm to the semi-streaming setting and removed the $\log n$ factor in the semi-streaming setting from the number
of passes to give an algorithm that finds an $(1-\eps)$-approximate maximum cardinality matching in $O\left(1/\eps^2\right)$ passes, where in each pass a maximal matching is found.
Furthermore, they showed implementations of their algorithm in the 
massively parallel computation (MPC) model, 
achieving the best-known bounds in both of 
these settings. In this paper, we extend their algorithm to other variants of the problem on bipartite graphs, 
including maximum weight matching and maximum cardinality $b$-matching and achieve 
novel improvements in a variety of \emph{scalable models}. The maximum cardinality $b$-matching problem (\mcbm) is a well-studied
generalization of \mcm. In \mcbm, each vertex is given an integer budget $b_v$ where each vertex can be matched to at most $b_v$ of their neighbors; 
a matching of maximum cardinality contains the maximum possible number of edges in the matching. The $b$-matching problem generalizes a number of 
real-life allocation problems such as server to client request serving, medical school residency matching, ad allocation, and many others.
Although the problem is similar to \mcm, often obtaining efficient algorithms for this problem requires non-trivial additional insights. As indicated in Ghaffari et al \cite{GGM22} $b$-matching problems can be considerably harder than matching.

\paragraph{Summary of Results}
In this paper, we specifically give the following results. Our auction algorithms and their analyses are described in detail in~\cref{sec:mwm} and~\cref{sec:b-matching}.
\begin{restatable}[Maximum Weight Bipartite Matching]{theorem}{mwmthm}\label{thm:mwm}
    There exists an auction algorithm for maximum weight bipartite matching (\mwm) that gives a $(1-\eps)$-approximation
    for any $\eps > 0$ and runs in $O\left(\frac{\log n}{\eps^8}\right)$ rounds of communication (with high probability) and 
    with $O\left(\log^2 n\right)$ bits per message. 
    This algorithm can be implemented in the
    multi-round, semi-streaming model using $O\left(n \cdot \log n \cdot \log(1/\eps)\right)$ space and $O\left(\frac{1}{\eps^8}\right)$ passes.
    This algorithm can be implemented in the work-depth 
    model in $O\left(\frac{m \cdot \log(n)}{\eps^7}\right)$ work
    and $O\left(\frac{\log^3(n)}{\eps^{7}}\right)$ depth. Finally, our algorithm can be implemented in the MPC model 
    using $O\left(\frac{\log\log n}{\eps^7}\right)$ rounds, $O(n \cdot \log_{(1/\eps)}(n)$ space per machine, 
    and $O\left(\frac{(n+m) \log(1/\eps) \log n}{\eps}\right)$ total space.
\end{restatable}
The best-known algorithms in the semi-streaming model for the maximum weight bipartite matching problem are the $(1/\eps)^{O(1/\eps^2)}$ 
pass, $O\left(n \poly(\log(n)) \poly(1/\eps)\right)$ space algorithm of Gamlath \etal~\cite{GKMS19} and the 
$O\left(\frac{\log(1/\eps)}{\eps^2}\right)$ pass, $O\left(\frac{n\log n}{\eps^2}\right)$ space algorithm of Ahn and Guha~\cite{AG11}. To the best of our knowledge,
our result is the first to achieve sub-polynomial dependence on $1/\eps$ in the space for the MWM problem in the semi-streaming model. Thus, we improve the 
space bound exponentially compared to the previously best-known algorithms in the streaming model. The 
best-known algorithms in the distributed and work-depth
models required $\poly(\log n)$ in the number of rounds
and depth, respectively~\cite{HS22} (for a large constant $c > 20$ in the exponent); in the 
MPC setting, the best previously known algorithms have exponential dependence on $\eps$~\cite{GKMS19}. We eliminate
such dependencies in our paper and our algorithm is also
simpler. A summary of previous
results and our results can be
found in~\cref{tab:results}.

\begin{restatable}[Maximum Cardinality Bipartite $b$-Matching]{theorem}{mcbmthm}\label{thm:b-matching}
    There exists an auction algorithm for maximum cardinality bipartite $b$-matching (\mcbm) that gives a $(1-\eps)$-approximation
    for any $\eps > 0$ and runs in $O\left(\frac{\log n}{\eps^2}\right)$ rounds of communication. This algorithm can be implemented in the
    multi-round, semi-streaming model using $O\left(\left(\sum_{i \in L} b_i + |R|\right)\log(1/\eps)\right)$ space and $O\left(\frac{1}{\eps^2}\right)$ passes. Our algorithm can be implemented in the shared-memory work-depth model 
    in $O\left(\frac{\log^3 n}{\eps^2}\right)$ depth 
    and $O\left(\frac{m \log n}{\eps^2}\right)$ total 
    work.
\end{restatable}

The best-known algorithms for maximum cardinality bipartite $b$-matching in the semi-streaming model 
is the $O\left(\frac{\log n}{\eps^3}\right)$ pass, $\tO\left(\frac{\sum_{i \in L \cup R} b_i}{\eps^3}\right)$ space algorithm of Ahn and Guha~\cite{AG11}.
In the general, non-bipartite setting (a harder setting than what we consider), a very recent $(1-\eps)$-approximation
algorithm of Ghaffari, Grunau, and Mitrovi\'c~\cite{GGM22} runs in $\exp\left(2^{O(1/\eps)}\right)$ passes and $\tO\left(\sum_{i \in L \cup R} b_i + \poly(1/\eps)\right)$
space. Here, we also improve the space exponentially in $1/\eps$ and, in addition, improve the number of passes
by an $O(\log n)$ factor. More details comparing our results to other related works are given in~\cref{sec:related} and~\cref{tab:results}. 

\begin{table}[htb!]
\renewcommand\arraystretch{2}
    \centering
    \scalebox{0.9}{
    \begin{tabular}{| c | c || c | c |c|c|}
         \hline
         \multicolumn{2}{|c||}{Model} &  \multicolumn{2}{|c||}{Previous Results }& \multicolumn{2}{|c||}{Our Results}  \\
         \hline
         \hline
          \multirow{2}*{\shortstack{Blackboard\\ Distributed}} & \mwm & $\Omega(n \log n)$ (trivial) & \cite{DNO14} & $O\left(\frac{n\log^3 (n)}{\eps^8}\right)$ & \cref{thm:distributed-rounds-alg}
          \\ \cline{2-6}
          & \mcbm &  
          $\Omega(nb\log n)$ & trivial & $O\left(\frac{nb \log^2 n}{\eps^2}\right)$ & \cref{lem:mcbm-distributed} \\
          \specialrule{.2em}{.1em}{.1em} 

          \multirow{2}*{Streaming} & \mwm &  
          \shortstack{$O\left(\frac{\log(1/\eps)}{\eps^2}\right)$ pass \\
          $O\left(\frac{n\log n}{\eps^2}\right)$ space} & \cite{AG11} & 
          \shortstack{$O\left(\frac{1}{\eps^8}\right)$ pass \\
          $O\left(n \cdot \log n \cdot \log(1/\eps)\right)$ space} & \cref{thm:streaming-reduce-passes}
           \\ \cline{2-6}
          & \mcbm & \shortstack{$O(\log n/\eps^3)$ pass \\ $\tO\left(\frac{\sum_{i \in L \cup R} b_i}{\eps^3}\right)$ space}
          & \cite{AG18} & \shortstack{$O\left(\frac{1}{\eps^2}\right)$ pass \\ $O\left(\left(\sum_{i \in L} b_i + |R|\right)\log(1/\eps)\right)$ space} & \cref{lem:b-semi-streaming}\\
          \specialrule{.2em}{.1em}{.1em} 

          \multirow{1}*{MPC} & \mwm &  
          \shortstack{$O_{\eps}(\log \log n)$ rounds \\ $O_{\eps}(n \poly(\log n))$\\ space p.m.} & \shortstack{\cite{GKMS19} \\(general)} & 
          \shortstack{$O\left(\frac{\log \log n}{\eps^7}\right)$ rounds\\ $O(n \cdot \log_{(1/\eps)}(n))$ space p.m.} & \cref{lem:weighted-mpc}
           \\ \cline{2-6}
          \specialrule{.2em}{.1em}{.1em} 

          \multirow{2}*{Parallel} & \mwm &  
          \shortstack{$O\left(m \cdot \poly\left(1/\eps, \log n\right)\right)$\\ work* \\
          $O\left(\poly\left(1/\eps, \log n\right)\right)$\\ depth*} & \shortstack{\cite{HS22} \\ (general)}& 
          \shortstack{$O\left(\frac{m\log(n)}{\eps^7}\right)$ work \\
          $O\left(\frac{\log^3 n}{\eps^{7}}\right)$ depth} & \cref{cor:shared-memory-parallel-reduced}
           \\ \cline{2-6}
          & \mcbm & N/A
          & N/A & \shortstack{$O\left(\frac{m \log n}{\eps^2}\right)$ work \\ $O\left(\frac{\log^3 n}{\eps^2}\right)$ depth} & \cref{lem:b-matching-parallel}\\
          \specialrule{.2em}{.1em}{.1em} 
    \end{tabular}}
    \caption{We assume the ratio between the largest weight edge and smallest weight edge in the graph is
    $\poly(n)$.
    Results for general graphs are labeled with (general); results that are specifically for bipartite
    graphs do not have a label. Upper bounds are given in terms of $O(\cdot)$ and lower bounds are given in
    terms of $\Omega(\cdot)$. ``Space p.m.'' stands for space per machine. The complexity measures for 
    the ``blackboard distributed'' setting is the total
    communication (over all rounds and players) in bits. $\poly(\log n, \eps)$ for the specified results
    indicated by * hides large constant factors in the exponents, specifically constants $c > 20$. Our results often
    exhibit a tradeoff of one complexity measure with another in our various models.}
    \label{tab:results}
\end{table}

\paragraph{Concurrent, Independent Work} In concurrent, independent work, Zheng and Henzinger~\cite{ZH23} study the 
maximum weighted matching problem in the sequential and dynamic
settings using auction-based algorithms. Their simple and elegant algorithm makes use of a sorted list of items (by utility) for 
each bidder and then matches the bidders one by one individually (in round-robin order) to their highest utility item. 
They also extend their algorithm to give dynamic results.
Due to the sequential nature of their matching procedure, they do not provide any results in scalable models such
as the streaming, MPC, parallel, or distributed models.

\subsection{Other Related Works}\label{sec:related}
There has been no shortage of work done on bipartite matching. In 
addition to the works we discussed in the introduction, there has
been a number of other relevant works in this general area of research. Here we discuss the additional works not discussed in~\cref{sec:intro}.
These include a plethora of results for $(1-\eps)$-approximate
maximum cardinality matching as well as some additional results for \mwm and $b$-matching. 
Most of these works use various methods 
to find augmenting paths with only a few works focusing on auction-based techniques. 
We hope that our paper further demonstrates the utility of auction-based approaches
as a type of ``universal'' solution across scalable models
and will lead to additional works in this area in the future. Although our work focuses on the bipartite matching problem,
we also provide the best-known bounds for the matching problem on general graphs here,
although this is a harder problem than our setting. We separate
these results into the bipartite matching results, the general
matching results, and lower bounds.

\paragraph{General Matching}
A number of works have considered \mcm in the streaming setting, providing
state-of-the-art bounds in this setting. Fischer \etal~\cite{FMU22} gave
a deterministic $(1-\eps)$-approximate \mwm algorithm in general graphs
in the semi-streaming model that uses $\poly(1/\eps)$ passes, improving
exponentially on the number of passes of Lotker \etal~\cite{LPPS15}.
Very recently, Assadi \etal~\cite{assadi2022semi} provided a semi-streaming
algorithm in optimal $O(n)$ space and $O\left(\log n \log (1/\eps)
/\eps\right)$ passes. They also provide a \mwm algorithm that also 
runs in $O(n)$ space but requires $\tilde{\Omega}(n/\eps)$ passes.
Please refer to these papers and references 
therein for older results in this area.
Ahn and Guha \cite{AG18} also considered the general 
weighted non-bipartite 
maximum matching problem in the semi-streaming model and 
utilize linear programming approaches for computing a $(2/3 - \eps)$-approximation and
$(1-\eps)$-approximation that uses $O(\log(1/\eps)/\eps^2)$ passes,
$O\left(n\cdot \left(\frac{\log(1/\eps)}{\eps^2} + \frac{\log n/\eps}{\eps}\right)\right)$ space,
and $O\left(\frac{\log n}{\eps^4}\right)$ passes, $O\left(\frac{n\log n}{\eps^4}\right)$ 
space, respectively. 

\paragraph{Bipartite Matching}
Ahn and Guha~\cite{AG18} also extended their results
to the bipartite \mwm and  
$b$-Matching settings with small changes. 
Specifically, in the \mwm setting, they 
give a $O(\log(1/\eps)/\eps^2)$ pass, 
$O(n\cdot ((\log(1/\eps))/\eps^2 + (\log n/\eps)/\eps))$ space algorithm. 
For maximum cardinality $b$-matching, they give a $O(\log n/\eps^3)$ pass
and $\tO\left(\frac{\sum_{i \in L \cup R} b_i}{\eps^3}\right)$ space 
algorithm. For exact bipartite \mwm in the semi-streaming model, 
Liu \etal~\cite{liu2020breaking} 
gave the first streaming algorithm to break the $n$-pass barrier in the 
exact setting; it 
uses $\tilde{O}(n)$ space and $\tilde{O}(\sqrt{m})$ passes using
interior point methods, SDD system solvers, 
and various other techniques 
to output the optimum matching with high probability. Work on  
bipartite \mwm prior to~\cite{liu2020breaking} either required 
$\Omega(n \log n)$ passes~\cite{JLS19} or only found approximate 
solutions~\cite{AG11, AG18,kapralov2013better}.

\paragraph{Lower Bounds}
Several papers have looked at matching problems from the
lower bound side. Konrad \etal~\cite{KRZ21} considered the communication 
complexity of graph problems in a blackboard model of computation (for 
which the simultaneous message passing model of Dobzinski \etal~\cite{DNO14}
is a special variant). Specifically, they show that any non-trivial graph
problem on $n$ vertices require $\Omega(n)$ bits~\cite{KRZ21} in communication
complexity. In a similar model called the \emph{demand query model}, Nisan
\cite{nisan2021demand} showed that any \emph{deterministic} 
algorithm that runs in 
$n^{o(1)}$ rounds where in each round at most $n^{1.99}$ demand queries are made,
cannot find a \mcm within a $n^{o(1)}$ factor of the optimum. This is in 
contrast to \emph{randomized} algorithms which can make such an approximation 
using only $O(\log n)$ rounds. For streaming matching algorithms, Assadi~\cite{Assadi22}
provided a conditional lower bound ruling out the possibilities of small 
constant factor approximations for two-pass streaming algorithms that solve the
\mcm problem. Such a lower bound also necessarily extends to \mwm 
and \mcbm. Goel~\etal~\cite{GKK12} provided a $n^{1+\Omega(1/\log\log n)}$
lower bound for the one-round
message complexity of bipartite $(2/3+\eps)$-approximate \mcm (this also 
naturally extends to a space lower bound). 
For older papers on these lower bounds, please refer to references
cited within each of the aforementioned cited papers. Finally, Assadi \etal~\cite{AKSY20} showed that any streaming algorithm that approximates 
\mcm requires either $n^{\Omega(1)}$ space or $\Omega(\log(1/\eps))$ passes.

\paragraph{Unweighted to Weighted Matching Transformations}
Current transformations for transforming unweighted to 
weighted matchings all either:
\begin{itemize}
    \item lose a factor of $2$ in the approximation factor~\cite{GP13,SW17}, or 
    \item increase the running time of the algorithm 
    by an exponential factor in terms of $1/\eps$, specifically, a factor of $\eps^{-O(1/\eps)}$~\cite{BDL21}. 
\end{itemize}
Thus, we cannot use such default transformations from 
unweighted matchings to weighted matchings
in our setting since all of the complexity measures
in this paper have only \emph{polynomial} dependence
on $\eps$ and all guarantee $(1-\eps)$-approximate
matchings. However, we do make use of \emph{weighted
to weighted matching transformations} provided our 
original weighted matching algorithms have only 
\emph{polylogarithmic} dependence on the maximum ratio
between edge weights in the graph. Such transformations
from weighted to weighted matchings do not increase the 
approximation factor and also allows us to \emph{eliminate the polylogarithmic dependence on 
the maximum ratio of edge weights}. 

\section{Preliminaries}\label{sec:prelims}

This paper presents algorithms for bipartite matching under various 
settings. The input consists of a bipartite
graph $G = (L \cup R, E)$. 
We denote the set of neighbors of any $i \in L, 
j \in R$ by $N(i), N(j)$, respectively. 
We present $(1-\eps)$-approximation algorithms where $\eps \in (0, 1)$ is 
our approximation parameter. All notations used in all of 
our algorithms in this paper are given in~\cref{table:mcbm-params}. The specified weight of an edge $(i,j)$ will become the valuation of the bidder $i$ for item $j$.

\begin{table}[htb]
    \centering
    \begin{tabular}{c l}
    \toprule
         Symbol & Meaning \\
         \midrule
         $\eps$ & approximation parameter \\
         $L, R$ & bidders, items, resp. wlog $|L| \leq |R|$\ \\
         $i, j, i', j'$ & $i \in L, j \in R, i' \in L', j' \in R'$, $i'$ (resp.\ $j'$) indicates copy of $i$ (resp.\ $j$)\\
         $p_{j}$ & current price of item $j$\\
         $D_{i}$ & demand set of bidder $i$ \\
         $(i, a_{i})$ & bidder $i \in L$ and currently matched item $a_{i}$ \\
         $o_i$ & the item matched to bidder $i$ in $\opt$\\
         $u_i$ & the utility of bidder $i$ which is calculated by $1 - p_{a_{i}}$\\
         $v_i(j)$ & the valuation of bidder $i$ for item $j$, i.e.\ the weight of edge $(i, j)$\\
         $C_i, C_j$ & copies of bidder $i \in L$, copies of item $j \in R$, resp.\ \\
         $L', R'$ & $L' = \bigcup_{i \in L} C_i$, $R' = \bigcup_{j \in R} C_j$ \\
         $E' , G'$ & $E' = \{(i^{(k)}, j^{(l)}) \mid (i, j) \in E, k \in [b_i], l \in [b_j]\}$, $G' = (L' \cup R', E')$ \\
         $b_i$ & cardinality of matching constraint for $i$ in $b$-matching\\
         $c_{i'}$ & price cutoff for bidder $i'$\\
         $W$ & ratio of the maximum weighted edge over the minimum weighted edge\\
         $G_d$ & induced subgraph consisting of $\left(\cup_{i \in L} D_{i}\right) \cup L$\\
         $\mmd$ & a \nonduplicate maximal matching in $G'_d$\\
         $M'_d, M_d$ & produced matching in $G'$, corresponding matching in $G$, resp.\ \\
         $M_{\max}$ & matching with largest cardinality produced\\
         \bottomrule
    \end{tabular}
    \caption{Table of Notations}
    \label{table:mcbm-params}
\end{table} 

\subsection{Scalable Model Definitions}
In addition, we consider a number of scalable models in our paper including
the \defn{blackboard distributed} model, the \defn{semi-streaming} model, the 
\defn{massively parallel computation (MPC)} model, and the \defn{parallel shared-memory 
work-depth} model.

\paragraph{Blackboard distributed model} We use the blackboard distributed model as defined
in~\cite{DNO14}. There are $n$ \emph{players}, one for each vertex of the left side of our 
bipartite graph (we assume wlog that the left side of the graph contains
fewer vertices). The players engage in a fixed communication protocol using messages sent to a central
coordinator. In other words, players write on a common ``blackboard.'' 
Players communicate using \emph{rounds} of communication where in each round the player sends a message
(of some number of bits) to the central coordinator.
Then, each player can receive a (not necessarily
identical) message in each round from the coordinator.
In every round, players choose to send messages
depending solely on the contents of the blackboard and their private information. Termination of the algorithm
and the final matching are determined by the central coordinator and the contents of the blackboard.
The measure of complexity is the number of rounds of the algorithm and the size of the message sent by each player
in each round.
One can also measure the total number of bits send by all messages by multiplying these two quantities.

\paragraph{Semi-streaming model} In this paper,
we use the semi-streaming model~\cite{feigenbaum2005graph} with arbitrarily ordered
edge insertions. Edges are arbitrarily (potentially adversarially) ordered
in the stream. For this paper, we only consider insertion-only streams. 
The space usage for semi-streaming algorithms is bounded by
$\tO(n)$. The relevant complexity measures in this model are the number of passes of the algorithm
and the space used.

\paragraph{Massively parallel computation (MPC) model} The massively parallel computation (MPC) model~\cite{beame2017communication,goodrich2011sorting,karloff2010model}
is a distributed model where different machines communicate with each other via a communication network. 
There are $M$ machines, each with $S$ space, and these machines communicate with each other using $Q$ rounds of 
communication. The initial graph is given in terms of edges and edges are partitioned arbitrarily across the machines. 
The relevant complexity measures are the total space usage ($M \cdot S$), space per machine $S$, and number of 
rounds of communication $Q$.

\paragraph{Parallel shared-memory work-depth model} The parallel shared-memory work-depth model~\cite{jaja1992introduction,ramachandran1990parallel,shiloach1982n2log} 
is a parallel model where different processors can process instructions in parallel and read and write from the 
same shared-memory. The relevant complexity measures for an algorithm in this model are the \emph{work} which is
the total amount of computation performed by the algorithm and the \emph{depth} which is the longest chain 
of sequential dependencies in the algorithm. 

\section{An Auction Algorithm for $(1-\eps)$-Approximate Maximum Weighted Bipartite
Matching}\label{sec:mwm}

We present the following auction algorithm for maximum (weighted) bipartite
matching (MWM) that is a generalization of the simple and elegant algorithm of 
Assadi \etal~\cite{ALT21} 
\ifcameraready
\else
(\cref{sec:mcm}) 
\fi
to the weighted setting. Our generalization requires several novel 
proof techniques and recovers the round guarantee of Assadi \etal~\cite{ALT21} in the maximum cardinality 
matching setting when the weights of all edges are $1$. 
Furthermore, we answer an open question posed by Dobzinski 
\etal~\cite{DNO14} for developing a $(1-\eps)$-approximation 
auction algorithm for maximum \emph{weighted} bipartite matching for which no prior algorithms
are known. Throughout this section, we denote the maximum ratio between two edge weights in the graph by $W$.
Our algorithm can also be easily extended into algorithms in various scalable models:
\begin{itemize}
\item a semi-streaming algorithm which uses $O\left(n \cdot \log n \cdot \log(1/\eps)\right)$ space ($n$ is the number of vertices in the bipartite 
graph) and which requires $O\left(\frac{1}{\eps^8}\right)$ passes,
\item a shared-memory parallel algorithm using 
$O\left(\frac{m \cdot \log(n)}{\eps^7}\right)$ work 
and $O\left(\frac{\log^3(n)}{\eps^{7}}\right)$ depth, and
\item an MPC algorithm using $O\left(\frac{\log\log n}{\eps^7}\right)$ rounds, $O(n \log_{(1/\eps)}(n))$ space per machine, 
and $O\left(\frac{(n+m) \log(1/\eps) \log n}{\eps}\right)$ total space. 
\end{itemize}

In contrast,
the best-known semi-streaming MWM algorithm of Ahn and Guha~\cite{AG11} requires $\tO(\log(1/\eps)/\eps^2)$ passes and $\tO\left(\frac{n\log n}{\eps^2}\right)$ space. 
Our paper shows a $O\left(\frac{1}{\eps^8}\right)$ pass algorithm that instead uses $O\left(n \cdot \log n \cdot \log(1/\eps)\right)$ space. Since $\eps = \Omega(1/n)$ (or otherwise
we obtain an exact maximum weight matching), our algorithm works in the semi-streaming model for all possible values of $\eps$ whereas 
Ahn and Guha~\cite{AG11} no longer works in semi-streaming when $\eps$ is small enough.

Our algorithm follows the general framework given in
\ifcameraready
\else
\cref{sec:mcm}.
\fi
However, both our algorithm and our analysis require
additional techniques. The main hurdle we must overcome is the fact that the
weights may be much \emph{larger} than the number of bidders and items.  
In that case, if we use the \mcm algorithm trivially in this setting where we increase the 
prices until they reach the maximum weight, 
the number of rounds can be very large, proportional to $\frac{w_{\max}}
{\eps^2}$ where $w_{\max}$ is the maximum weight of any edge. We avoid
this problem in our algorithm, instead obtaining only $\poly(\log(n))$
and $\eps$ dependence in the number of rounds.
Our main result in this section is the following (recall from~\cref{sec:intro}).

\mwmthm*

Before we give our algorithm, we give some notation used in this section.

\paragraph{Notation} 
The input bipartite graphs is represented by $G = (L \cup R,
E)$ where $L$ is the set of bidders and $R$ is the set of items.
Let $N(v)$ denote the neighbors of node $v \in L \cup R$. We 
use the notation $i \in L$ to denote bidders and $j \in R$ to
denote items.
For a bidder $i \in L$, the \emph{valuation} of $i$ for items 
in $R$ is defined as the function $v_i: R \rightarrow \mathbb{Z}_{\geq 0}$ where 
the function outputs a non-negative integer. If $v_i(j) > 0$, for any $j \in R$, then
$j \in N(i)$. Each bidder can match to at most one item. We
denote the bidder item pair by $(i, a_i)$ where $a_i$ is the
matched item and $a_i = \bot$ if $i$ is not matched to any item.
For any agent $i$ where $a_i \not= \bot$,
the \emph{utility} of a bidder $i$ given its matched item $a_i$ is 
$u_i \defined v_i(a_i) - p_{a_i}$ where $p_{a_i}$ is the current price of item $a_i$. 
For an agent $i$ where $a_i = \bot$, the utility of agent $i$ is $0$. %
We denote an optimum matching by $\opt$. We use the notation
$i \in \opt$ to denote a bidder who is matched in $\opt$
and $o_i$ to denote the item matched to bidder $i$ in $\opt$.

\paragraph{Input Specifications} 
In this section, we assume all weights are $\poly(n)$ 
where $n = |L| + |R|$.
We additionally assume the following characteristics 
about our inputs because we can perform a simple pre-processing of
our graph to satisfy these 
specifications. Provided an input graph $G = (L \cup R, E)$ with weights
$v_i(j)$ for every edge $(i, j) \in E$, we find the maximum weight among
all the weights of the edges, $w_{\max} = \max_{(i, j) \in E} 
\left(v_i(j)\right)$. 
We rescale the weights of all the edges
by $\frac{1}{w_{\max}}$ and remove all edges with rescaled 
weight $< \heps^{\ceil{\log_{(1/\heps)} (\min(m, W))} + 1}$. This 
upper bound of $\heps^{\ceil{\log_{(1/\heps)} (\min(m, W))} + 1}$
is crucial in our analysis.

In other words, 
we create a new graph $G' = (L \cup R, E')$ with the same set of bidders $L$
and items $R$. We associate the new weight functions $v'_i$ with each bidder $i \in L$
where $(i, j) \in E'$ if $v_i(j) \geq w_{\max} \cdot \heps^{\ceil{\log_{(1/\heps)} (\min(m, W))} + 1}$ and
$v'_i(j) = v_i(j)/w_{\max}$ for each $(i, j) \in E'$. Provided that finding the maximum weight edge 
can be done in $O(1)$ rounds in the blackboard distributed and MPC models, $O(1)$ passes in the streaming model,
and $O(n + m)$ work and $O(\log n)$ depth in the parallel model,
we assume the input to our algorithms is $G'$ (instead of the original graph $G$). The computation of $v_i'$ 
can be done on-the-fly as we run through our auction algorithm since every node knows $w_{\max}$.
In other words, we assume all inputs $G = (V, E)$ to our algorithm have scaled edge weights and 
$v_i(j)$ for $i \in L, j \in R$ are functions that return the scaled edge weights in the rest of this section. 

\subsection{Detailed Algorithm} 
We now present our auction algorithm for maximum weighted bipartite matching 
in~\cref{alg:auction}. The algorithm works as follows. 
Recall that we also assume the input to our algorithm is the scaled graph. %
This means that 
the maximum weight of the scaled edges is $1$ and there exists at least one edge
with weight $1$; hence, the maximum weight matching will have value at least $1$. 
We also initialize the tuples that keep track of matched items.
Initially, no items are assigned to bidders (\cref{mwm:assign-item}) and the 
prices of all items are set to $0$ (\cref{mwm:set-prices}). 

We perform $\phases$ phases of bidding (\cref{mwm:phases}).
In each phase, we form the \emph{demand set} $D_i$ of each unmatched bidder $i$. The demand
set is defined to be the set of items with non-zero utility which have \emph{approximately} the maximum 
utility value for bidder $i$ (\cref{mwm:unmatched-bidders,mwm:demand-price,mwm:demand}). 
This procedure is different from both \mcm and \mcbm (where no slack is needed in creating the demand set)
but we see in the analysis that we require this slack in the maximum utility value to ensure
that enough progress is made in each round.
Then, we create the induced subgraph
consisting of all unmatched bidders and their demand sets (\cref{mwm:subgraph}).
We find an arbitrary maximal matching in this created subgraph (\cref{mwm:matching}) by first finding the maximal matching in order
of decreasing buckets (from highest---bucket with the largest weights---to lowest). 
We partition the edges into buckets by their weight. An edge $(i, j)$ is in bucket $b$ if $\eps^{b-1} \leq v_i(j) < \eps^{b-2}$.
The ``highest'' bucket contains the largest weight edges and lower buckets contain smaller weight edges.
This means that we call our 
maximal matching algorithm $O(\log(W))$ times first on the induced subgraph consisting of the highest bucket, removing the 
matches, and then on the induced subgraph of the remaining edges plus the next highest bucket, and so on.
We use the folklore distributed maximal matching algorithm where in each round, a bidder 
uniformly-at-random picks a neighbor to match; this algorithm is also 
used in~\cite{DNO14} for the maximal matching step. This simple algorithm terminates in $O(\log n)$ 
rounds with high probability using $O(\log n)$ communication complexity. Such randomization is necessary 
to obtain $O(\log n)$ rounds using $O(\log n)$ communication complexity.

We rematch items according to the new matching 
(\cref{mwm:rematch-1,mwm:rematch-2}).
We then increase the price of each rematched item. The price increase depends on the 
weight of the matched edge to the item; higher weight matched edges have larger increases in price than smaller
weight edges. Specifically, the price is increased by $\heps \cdot v_i(a_i)$ where $v_i(a_i)$ is the weight of the 
newly matched edge between $i$ and $a_i$ (\cref{mwm:increase-price}).
The intuition behind this price increase is that we want to increase the price proportional to the weight 
gained from the matching since the price increase takes away from the overall utility of our matching. 
If not much weight is gained from the matching, then the price should not increase by much;
otherwise, if a large amount of weight is gained from the matching, then we can afford to increase the price by a larger amount.
We see later on in our analysis that this allows us to bucket the items according to their matched edge weight into 
$O\left(\ceil{\log_{(1/\eps)} (\min(m, W))}\right)$ buckets. Such bucketing is useful in ensuring that we have sufficiently many happy 
bidders with a sufficiently large total matched weight.
Finally, we return all matched items and bidders as our approximate matching and the
sum of the weights of the matched items %
as the approximate weight.
Obtaining the maximum weight of the matching in the original, unscaled graph is easy. We 
multiply the edge weights by $w_{\max}$ and the sum of these weights is the total 
weight of our approximate matching (\cref{mwm:return}).

\begin{algorithm}[!t]
    \caption{\label{alg:auction} Auction Algorithm for Maximum Weighted Bipartite
    Matching}
    \hspace*{\algorithmicindent} \textbf{Input:} A scaled graph $G = (L \cup R, E)$, parameter
    $0 < \heps < 1$, and the scaling factor $w_{\max}$. \\
    \hspace*{\algorithmicindent} \textbf{Output:} An $(1 - \mwmapprox)$-approximate 
    maximum weight bipartite matching. %
    \begin{algorithmic}[1]
        \State For each bidder $i \in L$, set $(i, a_i)$ to $a_i = 
        \bot$.\label{mwm:assign-item}
        \State For each item $j \in R$, set $p_j = 0$.\label{mwm:set-prices}
            \For{$d = 1, \dots, \phases$}\label{mwm:phases}
                \For{each unmatched bidder $i \in L$}\label{mwm:unmatched-bidders}
                    \State Let $U_i \defined \max_{j \in 
                        N(i), v_i(j) - p_j > 0} \left(v_i(j) - p_j\right)$.\label{mwm:demand-price}
                    \State Let $D_i \defined \{j \in R \mid p_j < v_i(j), v_i(j) - p_j \geq U_i - \heps \cdot v_i(j)\}$.\label{mwm:demand} %
                \EndFor
                \State Create the subgraph $G_d$ as the subgraph consisting of 
                $\left(\bigcup_{i \in L} D_i\right) \cup L$ and all edges.\label{mwm:subgraph}
                \State Find any arbitrary maximal matching $M_d$ of $G_d$ in order of highest bucket to lowest.\label{mwm:matching}
                \For{$(i, j) \in M_d$}\label{mwm:rematch-1}
                    \State Match $j$ to $i$ by setting $a_i = j$ and $a_{i'} = \bot$ 
                    for the previous owner $i'$ of $j$.\label{mwm:rematch-2}
                    \State Increase the price of $j$ to $p_j \leftarrow p_j + \heps \cdot v_i(j)$. \label{mwm:increase-price}
                \EndFor
                \State Let $M'$ be the matched edges in this current iteration.
            \EndFor
    \State Return the matching $M = \arg\max_{M'} \left(w_{\max} \cdot \sum_{i \in L}
    v_i(a_i)\right)$ as the approximate maximum weight matching and $(i, a_i) \in M$
    as the matched edges.\label{mwm:return}
    \end{algorithmic}
\end{algorithm}

\subsection{Analysis}\label{mwm:analysis}

In this section, we prove the approximation factor and round complexity of our algorithm. We use nearly the same definition of 
happy that is defined in~\cite{ALT21}.

\begin{definition}[$\eps$-Happy~\cite{ALT21}]\label{def:happy}
    A bidder $i$ is $\eps$-happy if $u_i \geq v_i(j) - p_j - \eps$ for every $j \in R$. 
\end{definition}

\begin{definition}[Unhappy]\label{def:unhappy}
A bidder $i$ is \defn{unhappy} at the end of round $d$ 
if they are unmatched and their demand set is non-empty.
\end{definition}

Note that a happy bidder may never be unhappy and vice versa. For this definition, we  assume that the demand set of a bidder can be computed at any point in time (not only when the algorithm computes it).

\paragraph{Approach} The main challenge we face in our \mwm analysis is that it is 
no longer sufficient to just show at least $(1-\eps)$-fraction of bidders in $\opt$ are happy in 
order to obtain the desired approximation.
Consider this simple example. Suppose a given instance has an optimum solution
$\opt$ with six matched bidders where one bidder is matched to an item via 
a weight-$1$ edge. It also has five additional bidders matched to items via
weight-$\frac{1}{\sqrt{n}}$ edges. Suppose we set $\eps = 1/6$ to be a constant.
Then, requiring $5/6$-fraction of the bidders in $\opt$ to be happy is not sufficient
to get a $5/6$-factor approximation. Suppose the five bidders matched with edges of 
weight $\frac{1}{\sqrt{n}}$ are the happy bidders. This is sufficient to satisfy the condition
that $5/6$-fraction of the bidders in $\opt$ are happy. However, the total combined weight
of the matching in this case is $\frac{5}{\sqrt{n}}$ while the weight of the optimum matching
is $\left(1 + \frac{5}{\sqrt{n}}\right)$. The returned matching then has weight smaller than a
$\frac{5}{\sqrt{n}}$-fraction of the optimum, and for large $n$, this is much less than the desired $5/6$-factor approximation.

Instead, we require a specific fraction of the total \emph{weight} of the optimum solution, $W_{\opt}$, to be 
matched in our returned matching. We ensure this new requirement by considering two types of 
unhappy bidders. Type $1$ unhappy bidders are bidders who are unhappy in round $k-1$ and remain
unmatched in round $k$. Type $2$ unhappy bidders are bidders who are unhappy in round $k - 1$
and become matched in round $k$. We show that there exists a round where the 
following two conditions are satisfied:

\begin{enumerate}
\item We bucket the bidders in $\opt$ according to the weight 
of their matched edge in $\opt$ such that bidders matched with similar weight edges are
in the same bucket; there exists a round where at most $(\eps^2)$-fraction
of the bidders in each bucket are Type $1$ unhappy.
\item We charge the weight a Type $2$ unhappy bidder $i$ obtains in round $k$ to $i$ in round $k -1$; there exists a round 
$k-1$ where a total of at most $\heps \cdot W_{\opt}$
weight is charged to Type $2$ unhappy bidders.
\end{enumerate}

Simultaneously satisfying both of the above conditions is enough to obtain our desired approximation. 
The rest of this section is devoted to showing our precise analysis using the above approach.

\paragraph{Detailed Analysis}
Recall that we defined the 
utility of agent $i$ to be the value of the item matched to her
minus its price $u_i = v_i(a_i) - p_{a_i}$. In this section, we use 
the definition of $\eps$-happy from~\cref{def:happy}.

A similar observation to the observation made in~\cite{ALT21} 
about the happiness of matched bidders can also be made in our case; however, since 
we are dealing with edge weights, we need to be careful to increment our
prices in terms of the newly matched edge weight. 
In other words, two different bidders 
could be $\eps_1$-happy and $\eps_2$-happy 
after incrementing the price of their 
respective items by $\eps_1$ and $\eps_2$ where $\eps_1 \neq \eps_2$; the incremented prices 
$\eps_1$ and $\eps_2$
depend on the matched edge weights of the items assigned to the bidders. We prove 
the correct happiness guarantees given by our algorithm below.

\begin{observation}\label{obs:happy}
At the end of every round, matched 
bidder $i$ with matched edge $(i, a_i)$, where $a_i$ is 
priced at $p_{a_i}$, is $\left(2\heps \cdot v_i(a_i)\right)$-happy.
At the end of every round, unmatched bidders with empty demand sets $D_i$ are $\heps$-happy. 
\end{observation}

\ifcameraready
\else
\begin{proof}
    Let $a_i$ be the item picked by $i$. 
    First, each matched bidder picks an item $a_i$ where $v_i(j) > p_j$ and it has utility at least $\max_{j \in N(i), p_j < v_i(j)}
    \left(v_i(j) - p_j\right) - \heps \cdot v_i(a_i)$. 
    Item $a_i$ increases its price by $\heps \cdot v_i(a_i)$ after it is picked
    by~\cref{mwm:increase-price} of~\cref{alg:auction}. 
    Thus, $u_i = v_i(a_i) - p_{a_i} \geq v_i(j) - p_{j} 
    - 2\heps \cdot v_i(a_i)$ for all $j \in N(i)$. This satisfies our given 
    happiness definition.
    
    Second, if an item remains matched to bidder $i$ that was previously 
    matched to $i$, then the item's price has not increased. Furthermore, since prices of items 
    are monotonically non-decreasing and 
    $\left(2\heps \cdot v_i(j)\right)$ is fixed for each
    edge $\{i, j\}$; 
    each bidder $i$ who was matched to 
    an item in a previous round 
    would remain $2\eps \cdot v_i(a_i)$-happy for the next round. 
    
    Finally, for all unmatched bidders with empty $D_i$, this means that allocating any item $j \in 
    N(i)$ to bidder $i$ results in $0$ gain in utility and hence 
    $u_i = 0 \geq v_i(j) - p_j > v_i(j) - p_j - \heps$ for all such bidders $i$ and $j \in N(i)$.
\end{proof}
\fi

For the weighted case, we need to consider what we call \defn{bidder weight buckets}. %
We define these weight buckets with respect to the optimum matching $\opt$. Recall our notation 
where $i \in \opt$ is a bidder who is matched in $\opt$ and $o_i$ is the matched item of the 
bidder in $\opt$. Bidder $i$ is in the $b$-th weight bucket if $\heps^{b - 1} \leq v_i(o_i) < \heps^{b-2}$. %

\begin{observation}\label{inv:weight-bucket}
All bidders $i \in \opt$ in bidder weight bucket $b$ satisfy $\heps^{b - 1} \leq v_i(o_i) < \heps^{b - 2}$.
\end{observation}

We now show that if a certain number
of bidders in $\opt$ are happy in our matching, then we obtain a matching with sufficiently large
enough weight. However, our guarantee is somewhat more intricate than the guarantee provided in~\cite{ALT21}.
We show that in $\phases$ rounds, there exists one round $d$
where a set of sufficient conditions are satisfied to obtain our approximation guarantee. To do this,
we introduce two types of unhappy bidders. Specifically, Type $1$ and Type $2$ unhappy bidders.

Each \emph{unhappy} bidder results in some loss of total matched weight. However, at the end of round $k-1$ it is 
difficult to determine the exact amount of weight lost to unhappy bidders. Thus, in our analysis, 
we determine, in round $k$, the amount of weight lost to unhappy bidders at the end of round $k-1$. The way that we determine
the weight lost in round $k - 1$ is by \emph{retroactively} categorizing an unhappy bidder in round $k - 1$ as a Type $1$ or Type $2$ 
unhappy bidder \emph{depending on what happens in round $k$}. Thus, for our analysis,
we categorize the bidders into categories of unhappy bidders for the \emph{previous} round. %

A \defn{Type $\mathbf{1}$} unhappy bidder in round $k - 1$ is a bidder $i$ that remains unmatched at the end of round $k$. 
In other words, a Type $1$ unhappy bidder was unhappy in round $k-1$
and either remains unhappy in round $k$ or becomes happy because it does not have any demand items anymore (and remains unmatched). 
A \defn{Type $\mathbf{2}$} unhappy bidder $i$ in round $k-1$ is a bidder who was unhappy in round $k - 1$ but is matched to an item in round $k$.
Thus, a Type $2$ unhappy bidder $i$ in round $k - 1$ becomes happy in round $k$ because a new item is matched to $i$.
Both types of bidders are crucial to our analysis given in the proof of~\cref{lem:eps-good-weight} since
they contribute differently to the potential amount of value that could be matched by our algorithm. Furthermore, the proof of~\cref{lem:weighted-eps-good} necessitates bounding the two quantities separately.

In the following lemma, let $\opt$ be the optimum matching in graph $G$ and $\wopt = \sum_{i \in \opt} v_i(o_i)$.
Let $B_b$ be the set of bidders $i \in \opt$ in bidder weight bucket $b$. 
If a Type $2$ unhappy bidder $i$
gets matched to $a_i$ in round $k$, we say the weight $v_i(a_i)$ is \defn{charged} to bidder $i$ in round $k - 1$.
We denote this charged weight as $c_i(a_i)$ when performing calculations for round $k-1$.

\begin{lemma}\label{lem:eps-good-weight}
Provided $G = (L \cup R, E)$ and an optimum weighted matching $\opt$ with weight $\wopt = \sum_{i \in \opt} v_i(o_i)$, 
if in some round $d$ of~\cref{mwm:phases} of~\cref{alg:auction} both of the following are satisfied,
\begin{enumerate}
    \item at most $\heps^2 \cdot |B_b|$ of the bidders in each bucket $b$ are Type $1$ unhappy and
    \item at most $\heps \cdot \wopt$ weight is charged to Type $2$ unhappy bidders,
\end{enumerate} 
then the matching in $G$ has weight at least $(1-6\heps) \cdot \wopt$. 
\end{lemma}

\begin{proof} %
    In such an iteration $r$, let $\happy$ denote the set of all happy bidders.
    For any bidder $i \in \happy \cap \opt$, by~\cref{def:happy} and~\cref{obs:happy}, 
    $u_i \geq v_i(o_i) - p_{o_i} - 2\heps \cdot v_i(a_i)$ 
    where $o_i$ is the item matched to $i$ in $\opt$ and $a_i$ is the item matched to $i$ 
    from our matching.
    
    Before we go to the core of our analysis, we first make the observation below that we can, in general, disregard prices of the
    items in our analysis. Let $M$ be our matching.
    The sum of the utility of every matched bidder in our matching can be upper and lower bounded by the 
    following expression:
    
    \begin{align*}
        \sum_{i \in M} \left(v_i(a_i) - p_{a_i}\right) \geq \sum_{i \in \opt \cap \happy} u_i \geq \sum_{i \in \opt \cap \happy} \left(v_i(o_i) - p_{o_i} - 2\heps \cdot v_i(a_i)\right). 
    \end{align*}
    
    As in the maximum cardinality matching case, all items with non-zero price are matched to a bidder. We can then simplify the above expression to give
    
    \begin{align}
        \sum_{i \in M} v_i(a_i) - \sum_{j \in R} p_{j} &\geq \sum_{i \in \opt \cap \happy} v_i(o_i) - \sum_{i \in \opt \cap \happy}p_{o_i} - \sum_{i \in \opt \cap \happy} 2\heps \cdot v_i(a_i)\label{eq:first-simplify}\\
        \sum_{i \in M \setminus (\opt \cap \happy)} v_i(a_i) &+ \sum_{i \in \opt \cap \happy} (1+2\heps) v_i(a_i) 
        - \sum_{j \not\in \{o_i \mid i \in \opt \cap \happy\}} p_j \geq \sum_{i \in \opt \cap \happy} v_i(o_i) \label{eq:price-before-simplified}\\
        \sum_{i \in M} (1+2\heps)v_i(a_i) &- \sum_{j \not\in \{o_i \mid i \in \opt \cap \happy\}} p_j \geq \sum_{i \in \opt \cap \happy} v_i(o_i).\label{eq:prices-gone-simplified}
    \end{align}

    \cref{eq:first-simplify} follows from the fact that all non-zero priced items are matched.
    \cref{eq:price-before-simplified} follows from separating $\opt \cap \happy$ 
    from the left hand side and moving the summation of the 
    $2\eps \cdot v_i(a_i)$ values over $\opt \cap \happy$ from the right hand side to the left hand side. 
    Finally, \cref{eq:prices-gone-simplified} follows because $\sum_{i \in M} (1+2\heps)v_i(a_i)$ 
    upper bounds the left hand side expression for $\sum_{i \in M \setminus (\opt \cap \happy)} v_i(a_i) 
    + \sum_{i \in \opt \cap \happy} (1+2\heps) v_i(a_i)$.

    Let $\unhappy_{1}$ denote the set of Type $1$ unhappy bidders and $\unhappy_{2}$ denote the 
    set of Type $2$ unhappy bidders. We let $c_i(a_i)$ be the weight charged to bidder $i$ in $\unhappy_2$ in the next round. Recall that
    each bidder in $\unhappy_2$ is matched in the next round.
    
    For each bucket, $b$, we can show the following using our assumption that at most $\eps^2 \cdot |B_b|$ of the bidders in bucket $b$ are Type $1$ unhappy,

    \begin{align}
        \sum_{i \in B_b \cap \happy} v_i(o_i) &\geq \sum_{i \in B_b \setminus \unhappy_2} v_i(o_i) - \eps^2 \cdot \eps^{b-2} \cdot |B_b| \label{eq:bucket-1}\\
        &\geq \sum_{i \in B_b \setminus \unhappy_2} v_i(o_i) - \eps \cdot \eps^{b-1} \cdot |B_b|\label{eq:bucket-2}\\
        &\geq \sum_{i \in B_b \setminus \unhappy_2} v_i(o_i) - \sum_{i \in B_b} \eps \cdot v_i(o_i).\label{eq:bucket-3}
    \end{align}
    \cref{eq:bucket-1} shows that one can lower bound
    the sum of the optimum values of all happy bidders in bucket $b$ by the sum of the optimum values of 
    all bidders who are not Type-2 unhappy minus some factor. First, $\sum_{i \in B_b \setminus \unhappy_2} v_i(o_i)$ 
    is the sum of the optimum values of all bidders in bucket $b$ \emph{except} for the Type-2 unhappy bidders. 
    Now, we need to subtract the maximum sum of values given to the Type-1 unhappy bidders. We know that bucket
    $b$ has at most $\eps^2 \cdot |B_b|$ Type-1 unhappy bidders. Each of these bidders could be assigned
    an optimum item with value at most $\varepsilon^{b-2}$ (by~\cref{inv:weight-bucket}).
    Thus, the maximum value lost to Type-1 unhappy bidders is $\varepsilon^2 \cdot 
    \varepsilon^{b-2} \cdot |B_b|$, leading to
    \cref{eq:bucket-1}. Thus, the maximum value of weight 
    lost to all Type-$1$ unhappy bidders in bucket $b$ is $\eps \cdot \eps^{b-1} \cdot |B_b|$. Then,~\cref{eq:bucket-3} follows because 
    $v_i(o_i) \geq \eps^{b-1}$ for all $i \in B_b$. This means that $\sum_{i \in B_b} v_i(o_i) 
    \geq \eps^{b-1} \cdot |B_b|$. 

    Summing~\cref{eq:bucket-3} over all buckets $b$ we obtain

    \begin{align}
        \sum_{i \in \opt \cap \happy} v_i(o_i) &\geq \sum_{i \in \opt \setminus \unhappy_2} v_i(o_i) - \sum_{i \in \opt} \eps \cdot v_i(o_i).\label{eq:simplified-type-1}
    \end{align}
    
    We now substitute our expression obtained in~\cref{eq:simplified-type-1} into~\cref{eq:prices-gone-simplified},

    \begin{align}
        \sum_{i \in M} (1+2\heps) v_i(a_i) - \sum_{j \not\in \{o_i \mid i \in \opt \cap \happy\}} p_j &\geq  \sum_{i \in \opt \setminus \unhappy_2} v_i(o_i) - \sum_{i \in \opt} \eps \cdot v_i(o_i).\label{eq:all-type-1}
    \end{align}

    The last thing that we need to show is a bound on the weight lost due to bidders in $\opt \cap \unhappy_2$. 
    We now consider our second assumption which states that at most $\eps \cdot W_{\opt}$ weight is charged to Type $2$ unhappy bidders.
    Since all bidders $i \in \unhappy_2$ become happy in the next round, we can bound the weights charged to the Type-$2$ unhappy bidders using~\cref{obs:happy}
    by

    \begin{align}
        \sum_{i \in \opt \cap \unhappy_2} c_i(a_i) \geq \sum_{i \in \opt \cap \unhappy_2} \left(v_i(o_i) - p_{o_i} - 2\eps \cdot c_i(a_i)\right).\label{eq:type-2-initial}
    \end{align}

    Note first that $\sum_{j \not\in \{o_i \mid i \in \opt \cap \happy\}} p_j 
    \geq \sum_{i \in \opt \cap \unhappy_2} p_{o_i}$ since $\opt \setminus \left(\opt \cap \happy\right)$ 
    includes $\opt \cap \unhappy_2$ 
    so we can remove the prices from these bounds in~\cref{eq:final-1}.
    We add~\cref{eq:type-2-initial} to~\cref{eq:all-type-1} and use our assumptions to obtain

    \begin{align}
        \sum_{i \in M} (1+2\heps) v_i(a_i) &+ \sum_{i \in \opt \cap \unhappy_2} c_i(a_i) \geq \sum_{i \in \opt} (1-\eps) \cdot v_i(o_i) - \sum_{i \in \opt \cap \unhappy_2} 2\eps \cdot c_i(a_i)\label{eq:final-1}\\
        \sum_{i \in M} (1+2\heps) v_i(a_i) &\geq \sum_{i \in \opt} (1-\eps) \cdot v_i(o_i) - \sum_{i \in \opt \cap \unhappy_2} (1+2\eps) \cdot c_i(a_i)\label{eq:final-2}\\
        &\geq \left(\sum_{i \in \opt} (1-\eps) \cdot v_i(o_i)\right) - (1+2\eps)\cdot \eps \cdot \wopt\label{eq:final-3}\\
        \sum_{i \in M} v_i(a_i) &\geq \frac{(1-4\eps)}{(1+2\heps)} \cdot \sum_{i \in \opt} v_i(o_i)\label{eq:final-4}\\
        \sum_{i \in M} v_i(a_i) &\geq (1-6\eps)\wopt.\label{eq:final-approx}
    \end{align}

    \cref{eq:final-1} follows from summing $\sum_{i \in \opt} (1-\eps) \cdot v_i(o_i) = \sum_{i \in \opt \setminus \unhappy_2} v_i(o_i) + \sum_{i \in \opt \cap \unhappy_2} v_i(o_i) - \sum_{i \in \opt} \eps \cdot v_i(o_i) = \sum_{i \in \opt} (1-\eps) \cdot v_i(o_i)$.
    \cref{eq:final-2} follows from moving $\sum_{i \in \opt\cap \unhappy_2} c_i(a_i)$ to the right hand side. 
    \cref{eq:final-3} follows from substituting our assumption that $\sum_{i \in \opt \cap \unhappy_2} c_i(a_i) \leq \eps \cdot \wopt$. 
    \cref{eq:final-4} follows from simple manipulations and since $\wopt = \sum_{i \in \opt}v_i(o_i)$. Finally, 
    \cref{eq:final-approx} follows because $\frac{(1-4\eps)}{(1+2\heps)} \geq (1-6\eps)$ for all $\eps > 0$ and
    gives the desired approximation given in the lemma statement.

\end{proof}

We show that the conditions of~\cref{lem:eps-good-weight} are satisfied for at least one round if the algorithm is run for at least
$\ceil{\mwmpretransform}$ rounds. We prove this using potential functions similar to the potential functions used for \mcm. We first bound the maximum
value of these potential functions.

\begin{lemma}\label{lem:upper-bound}
Define the potential function $\Phi_{items} \defined \sum_{j \in R} p_{j}$. Then the upper bound for this potential is $\Phi_{items} \leq \wopt$.
\end{lemma}

\begin{proof} %
We show that the potential function $\Phi_{items}$ is always upper bounded by $\wopt$ 
via a simple proof by contradiction. Suppose that $\Phi_{items} > 
\wopt$, then, we show that the matching obtained by our algorithm has 
weight greater than $\wopt$, a contradiction. For a bidder/item pair, $(i, a_i)$, the weight of edge $(i, a_i)$
is at least $p_{a_i} - 2\eps \cdot v_i(a_i)$. Let $p'_{a_i}$ be the price of $a_i$ before the last reassignment of $a_i$
to $i$. Furthermore, since $i$ picked $a_i$, it must mean that $v_i(a_i) > p'_{a_i}$ since $a_i$ would not be included in $D_i$ otherwise. 
This means that the sum of the weights of all the matched edges is at least $\sum_{(i, a_i)} v_i(a_i) >
\sum_{(i, a_i)} p'_{a_i} \geq \Phi_{items} > \wopt$ by our assumption that $\Phi_{items} > \wopt$. 
Thus, we obtain that we get a matching with greater weight than the optimum weight matching, a contradiction.
\end{proof}
    
\begin{lemma}\label{lem:weighted-eps-good}
There exists a phase $d \leq \mwmpretransform$ wherein \emph{both} of the following statements are satisfied:

\begin{enumerate}
    \item At most $\heps^2 \cdot |B_b|$ bidders in bucket $b$ are Type-$1$ unhappy for all buckets $b$;
    \item The set of all Type-$2$ unhappy bidders results in a loss of less than $\heps \cdot \wopt$ weight (in charged weight) where $\wopt$ is
    the optimum weight attainable by the matching.
\end{enumerate} 

Recall that we assign each bidder to a weight bucket using the weight assigned to the bidder in $\opt$. 
\end{lemma}

\ifcameraready
\else
\begin{proof} %
    We use similar potential functions to the proof of Lemma 2.2 in~\cite{ALT21} for
    each bucket $b$ but our argument is more intricate.
    First, the potential functions do not both start at $0$. 
    Specifically, we have a separate potential function for each bucket $b$, $\Phi_{bidders, b}$ as well as a potential function on all 
    the prices of the items, $\Phi_{items}$:
    
    \begin{align*}
        &\Phi_{bidders, b} \defined \sum_{i \in B_b} \max_{j \in N(i)} (v_i(j) - p_j, 0)\\
        &\Phi_{items} \defined \sum_{j \in R} p_{j}.
    \end{align*}
    
    The first one bounds the sum of the maximum utility of the bidders in $\opt$ and in bucket $b$
    and the second one bounds the sum of the prices of all items in $R$. 
    We have $0 \leq \Phi_{bidders, b} \leq |B_b|$ for all valid $b$. %
    The maximum possible utility obtained from each item is at most $1$ because 
    the weight of any edge is at most $1$. There are at most $|B_b|$ items in bucket $b$ so the maximum possible utility is $|B_b|$. $ $%
    Now, we argue that the minimum value of $\Phi_{bidders, b}$ is $0$. The minimum value of the 
    expression $\max_{j \in N(i)}\left(v_i(j) - p_j, 0\right)$ is $0$. 
    Thus, the sum of the expressions for all bidders in $B_b$ is at least $0$. %
    We also have $0 \leq \Phi_{items} \leq \wopt$ as we proved in~\cref{lem:upper-bound}.

    We consider \emph{slots} in increasing/decreasing our potential functions. We consider the slots to be the \emph{maximum} number of times 
    a particular price for an item $j$ can increase before it becomes $\geq 1$. By this definition, there are a total of
    $\left(\log_{(1/\heps)}\left(W\right) + 2\right) \cdot \frac{1}{\heps}$ slots for each item $j \in R$.
    This is due to the fact that there are at most
    $\ceil{\log_{(1/\eps)}(W)} + 2$ buckets provided 
    that we removed all edges with weight 
    less than $\eps^{\ceil{\log_{(1/\eps)}(\min(m, W))} + 1}$. For each bucket, the price can increase at most $1/\eps$ times
    before it becomes too large and can no longer be increased by any edge with weight in that bucket. This results in the maximum number of slots
    per item being upper bounded by $\left(\log_{(1/\heps)}\left(W\right) + 2\right) \cdot \frac{1}{\heps}$.
    We say that a bidder increasing the price of an item as \emph{taking one slot} from $\Phi_{bidders, b}$ or $\Phi_{items}$.
    Since $\Phi_{bidders, b}$ is monotonically non-increasing and $\Phi_{items}$ is monotonically non-decreasing, once a slot is filled, it cannot
    become free again.
    
    We first show that Type-$1$ unhappy bidders in bucket $b$ take at least one slot each from $\Phi_{bidders, b}$ for each round they are unhappy. 
    That is, we show that the increase in price is at least equal to $\heps^{b}$ 
    where $b$ is the smallest bucket $j \in D_i$ is in for each Type-$1$ unhappy bidder $i$. This is the case since we match edges from largest to smallest weight; hence, if a bidder $i$ is unmatched, then all of the items in $D_i$ are matched to bidders with edge weights in the same or 
    higher buckets.
    The smallest bucket that $j \in D_i$ is in is given by $U_i/(1+\eps)$ since in order for $j$ to be included in $D_i$, it must
    be the case that $(1+\eps)v_i(j) \geq U_i$ so $v_i(j) \geq U_i/(1+\eps)$. This means that $U_i$ decreases by 
    at least $(\eps U_i)/(1+\eps)$. Suppose that $U_i$ is in bucket $b$ then $U_i \geq \eps^{b-1}$. %
    By this argument, it can use at most $\frac{\left(\eps^{b-2}-\eps^{b-1}\right)}{\left(\frac{\eps \cdot \eps^{b-1}}{(1+\eps)}\right)} = \frac{1 - \eps^2}{\eps^2}
    \leq \frac{1}{\eps^2}$ slots in bucket $b$. 
    Provided the number of buckets is upper bounded by $\log_{(1/\heps)}(W) + 2$,
    each unhappy Type $1$ bidder uses at most $\frac{\log_{(1/\heps)}(W) + 2}{\heps^2}$ slots before their demand 
    set becomes empty. Then, at most $\frac{|B_b|(\log_{(1/\heps)}(W) + 2)}{|B_b|\heps^4} = \frac{\log_{(1/\heps)}(W) + 2}{\heps^4}$ rounds exist where
    $\geq \heps^2 \cdot |B_b|$ bidders in bucket $b$ are Type-$1$ unhappy and $\Phi_{bidders, b} > 0$.
    
    We now consider Type-$2$ unhappy bidders. Let the item $j'$ matched to $i$ in round $k$ be the \emph{charged} item to Type-$2$ unhappy bidder $i$ 
    and $c_i(j')$ be $i$'s \emph{charged weight} in round $k - 1$. Suppose that round $k - 1$ has $\geq \heps \cdot \wopt$ charged weight where
    $\wopt$ is the optimum weight. Noticeably, we \emph{charge}
    the item that is matched to $i$ in round $k$ to $i$ in round $k - 1$. Thus, in round $k$, the total increase in $\Phi_{items}$ is 
    at least $\heps \cdot \eps \cdot \wopt = \eps^2 \cdot \wopt$ assuming the charged weight is at least $\eps \cdot \wopt$. 
    Thus, in $\frac{\wopt}{\heps^2 \wopt} \leq \frac{1}{\heps^2}$ rounds, there exists at least one 
    round where $< \heps \cdot \wopt$ weight (in charged weight) is lost by Type-$2$ unhappy bidders. 
    
    The final observation that remains is that an unhappy bidder $i$ in round $k - 1$ \emph{must} either be Type-$1$ or Type-$2$ unhappy. 
    This is true since $i$ must be either matched or unmatched in round $k$. Thus, the unhappy 
    bidder contributes to at least one of the potential functions. By our argument above, a total of $\frac{\log_{(1/\heps)}(W) + 2}{\heps^4}$ rounds can 
    exist where $\geq \heps \cdot |B_b|$ bidders are Type $1$ unhappy in bucket $b$. Furthermore, also by what we showed above, there exists at most $\frac{1}{\heps^2}$
    rounds where Type $2$ unhappy bidders contribute $\geq \heps \cdot \wopt$ weight to $\Phi_{items}$. 
    There are $O(\log(W))$ buckets where for each bucket at most $\frac{\log_{(1/\heps)}(W) + 2}{\heps^4}$ total rounds 
    can exist where $\geq \eps \cdot |B_b|$ bidders are Type-$1$ unhappy. 
    Thus, by the pigeonhole principle, in $\frac{2\left(\log_{(1/\heps)}^2(W) + 2\right)}{\heps^4}$ phases, 
    both conditions will be satisfied.
\end{proof}
\fi

Using the above lemmas, we can prove our main theorem that our algorithm gives a 
$(1-7\eps)$-approximate maximum weight bipartite matching in $O\left(\frac{\log^3(W) \cdot \log(n)}{\eps^4}\right)$ 
distributed rounds using $O(\log n)$ communication complexity. 

\begin{theorem}\label{thm:main}
    \cref{alg:auction} returns a $(1-7\eps)$-approximate maximum weight bipartite matching $M$
    in $O\left(\frac{\log^3 (W) \cdot \log n}{\eps^4}\right)$ rounds \whp using $O(\log n)$ bits of communication per message 
    in the broadcast model.
\end{theorem}

\ifcameraready
\else
\begin{proof}
    In each of the $O\left(\frac{\log^2(W)}{\eps^4}\right)$ phases given in~\cref{alg:auction}, we run the distributed maximal
    matching algorithm $O(\log(W))$ times using $O(\log(W)\cdot \log n)$ rounds, resulting in a total of $O\left(\frac{\log^3 (W) \cdot \log n}{\eps^4}\right)$ rounds.
    Then, the communication complexity is determined by what bidders write on the blackboard. In each phase, finding the maximal matching
    requires $O(\log n)$ communication since bidders compute their $D_i$ individually and picks a neighbor uniformly at random to match.
    Then, to update the prices of the items, each bidder sends at most $O(\log_{(1/\eps)}(m)) = O(\log n)$ bits for the bidding price. 
\end{proof}
\fi

\paragraph{Reducing the Round Complexity}

We can use the following transformation from Gupta-Peng~\cite{GP13} to reduce the round complexity 
at an increase in the communication complexity. For completeness, we give the theorem for the transformation
\ifcameraready
\else
in~\cref{app:guptapeng}. 
\fi

\begin{theorem}\label{thm:distributed-rounds-alg}
    There exists a $(1-\eps)$-approximate distributed algorithm for maximum weight bipartite matching that runs 
    in either:
    \begin{itemize}
    \item $O\left( \frac{\log n}{\eps^7}\right)$ rounds of communication using $O\left(\frac{\log^2 n}{\eps}\right)$
    bits of communication, or
    \item $O\left(\frac{\log n}{\eps^8}\right)$ rounds of communication using $O\left(\log^2 n\right)$
    bits of communication
    \end{itemize}
    where we assume the maximum ratio between weights of edges in the input graph is $\poly(n)$.
    In the blackboard model, this requires 
    a total of $O\left(\frac{n\log^3 n}{\eps^8}
    \right)$ bits of communication.
\end{theorem}

\ifcameraready
\else
\begin{proof}
 This follows from applying~\cref{thm:distributed-transformation} to~\cref{alg:auction} with bounds given
 by~\cref{thm:main}. We define $f(\eps) = \eps^{-O(\eps^{-1})}$ as in~\cite{GP13} (reconstructed
 in~\cref{app:guptapeng}).
\end{proof}
\fi

\subsection{Semi-Streaming Implementation}

The implementation of this algorithm in the semi-streaming model is very similar to the implementation of the 
\mcm algorithm of Assadi \etal~\cite{ALT21}.

\begin{lemma}\label{lem:mwm-streaming}
    Given a weighted graph $G = (V, E)$ as input in an arbitrary edge-insertion stream where all weights are at most $\poly(n)$,
    there exists a semi-streaming algorithm which uses $O\left(\frac{\log^3(W)}{\eps^4}\right)$ passes and $O\left(n \cdot \log(1/\eps)\right)$ space that computes a
    $(1 - \eps)$-approximate maximum weight bipartite matching for any $\eps > 0$. 
\end{lemma}

\ifcameraready
\else
\begin{proof} %
    We implement~\cref{alg:auction} in the semi-streaming model as follows. We use one pass to determine $w_{\max}$, the set of bidders, and the set of 
    items. We initialize all variables to their respective initial values.
    Then for each round, we make two passes. In the first pass, we compute $U_i$ for each bidder. Then, in the second pass, 
    we greedily find a maximal matching. We do not store $D_i$ in memory. Instead, we compute $D_i$ as we see the edges
    and for each edge that connects a bidder $i$ to an item $j$ in $D_i$ and $i$ is not newly matched this round, we match $i$ and $j$.
    This only requires $O(n)$ space to perform this matching. We store $M_d$, computed in this manner, in memory in $O(n)$ space. Then,
    using our stored $M_d$, we increase the price of each newly matched item. We can store all prices of items in $O(n \log(1/\eps))$ memory
    assuming the weights were originally at most $\poly(n)$.
    Finally, we store the matching after the current round and the maximum weight matching from previous rounds. 
    Returning the stored matching does not require additional space. Altogether, we use $O(n \log(1/\eps))$ space.
\end{proof}
\fi

\paragraph{Reducing the Number of Passes}

We use the transformation of~\cite{GP13} as stated 
\ifcameraready
\else
in~\cref{app:guptapeng}
\fi
to eliminate our dependence on $n$ within our number of rounds. 
The transformation is as follows. For each instance of $(1+\eps)$-MWM, we maintain the prices in our algorithm
for each of the nodes involved in each of the copies of our
algorithm. When an edge arrives in the stream, we first partition it into
the relevant level of the appropriate copy of the structure. 

\begin{theorem}\label{thm:streaming-reduce-passes}
    There exists a $(1-\eps)$-approximate streaming algorithm for maximum weight bipartite matching that uses
    $O\left(\frac{1}{\eps^8}\right)$ passes in $O\left(n \cdot \log n \cdot \log(1/\eps)\right)$
    space.
\end{theorem}

\ifcameraready
\else
\begin{proof}
    We apply~\cref{thm:streaming-transformation} to~\cref{lem:mwm-streaming} with $f(\eps) = \eps^{-O(1/\eps)}$.
\end{proof}
\fi

\subsection{Shared-Memory Parallel Implementation}

The implementation of this algorithm in the shared-memory work-depth model follows almost directly from our 
auction algorithm. We show the following lemma when directly implementing our auction algorithm.

\begin{lemma}\label{lem:work-depth-complexity}
    Given a weighted graph $G = (V, E)$ as input where all weights are at most $\poly(n)$,
    there exists a shared-memory parallel algorithm which uses $O\left(\frac{m\log^3(W)}{\eps^4}\right)$ work
    and $O\left(\frac{\log^3(W) \cdot \log^2(n)}{\eps^4}\right)$ depth that computes a
    $(1 - \eps)$-approximate maximum weight bipartite matching for any $\eps > 0$. 
\end{lemma}

\ifcameraready
\else
\begin{proof}
To implement our auction algorithm in the shared-memory parallel model, the only additional procedure we require
is a maximal matching algorithm in the shared-memory parallel model. The currently best-known maximal 
matching algorithm uses $O\left(m\right)$ work and $O\left(\log^2n\right)$ depth~\cite{BFS12,FischerN18,Birn2013}. 
Combined with our auction algorithm, we 
obtain the work and depth as desired in the statement of the lemma.
\end{proof}
\fi

Using the transformations, we can reduce the depth of our shared-memory parallel algorithms. 

\begin{theorem}\label{cor:shared-memory-parallel-reduced}
    Given a weighted graph $G = (V, E)$ as input where all weights are at most $\poly(n)$,
    there exists a shared-memory parallel algorithm which uses $O\left(\frac{m\log(n)}{\eps^7}\right)$ work
    and $O\left(\frac{\log^3(n)}{\eps^{7}}\right)$ depth that computes a
    $(1 - \eps)$-approximate maximum weight bipartite matching for any $\eps > 0$. 
\end{theorem}

\ifcameraready
\else
\begin{proof}
We apply~\cref{thm:parallel} to~\cref{lem:work-depth-complexity} with $f(\eps) = \eps^{-O(1/\eps)}$. 
\end{proof}
\fi

\subsection{MPC Implementation}

We implement our auction algorithm in the MPC model below. 

\begin{lemma}\label{lem:mpc-weighted}
Given a weighted graph $G = (V, E)$ as input where all weights are at most $\poly(n)$,
there exists a MPC algorithm using $O\left(\frac{\log^3\left(W\right) \cdot \log\log n}{\eps^4}\right)$ 
rounds, $O(n)$ space per machine, 
and $O\left(n\log(1/\eps) + m\right)$ total space that computes a $(1-\eps)$-approximate maximum 
weight bipartite matching for any $\eps > 0$. 
\end{lemma}

\ifcameraready
\else
\begin{proof}
    As in~\cite{ALT21}, we can use standard MPC techniques (sorting and prefix sum computation) to compute 
    $U_i$, $D_i$, and create the resulting subgraph
    in $O(1)$ rounds and $O(n)$ memory per machine in each phase of~\cref{alg:auction}.
    The MPC algorithm for computing maximal matchings require $O(\log \log n)$ rounds and $O(n)$ memory per machine~\cite{BHH19}.
\end{proof}
\fi

As before, we can improve the complexity of our MPC algorithm using the transformations 
\ifcameraready
\else
in~\cref{app:guptapeng}.
\fi

\begin{theorem}\label{lem:weighted-mpc}
    Given a weighted graph $G = (V, E)$ as input where all weights are at most $\poly(n)$,
    there exists a MPC algorithm using $O\left(\frac{\log \log n}{\eps^7}\right)$ rounds, $O(n \cdot \log_{(1/\eps)}(n))$ space per machine, 
    and $O\left(\frac{(n+m) \log(1/\eps) \log n}{\eps}\right)$ total space that computes a $(1-\eps)$-approximate maximum 
    weight bipartite matching for any $\eps > 0$. 
\end{theorem}

\ifcameraready
\else
\begin{proof}
    We apply~\cref{thm:mpc} to~\cref{lem:mpc-weighted} with $f(\eps) = \eps^{-O(1/\eps)}$.
\end{proof}
\fi

\section{A $(1-\eps)$-approximation Auction Algorithm for $b$-Matching}\label{sec:b-matching}

We show in this section that we also obtain an auction-based algorithm for \mcbm by 
extending the auction-based algorithm of~\cite{ALT21}. This algorithm also leads to better streaming
algorithms for this problem.
We use the techniques introduced in the auction-based 
\mcm algorithm of Assadi, Liu, and Tarjan~\cite{ALT21} 
\ifcameraready
\else
(discussed in~\cref{sec:mcm}) 
\fi
as well as new techniques developed in this section to obtain a $(1-\eps)$-approximation
algorithm for bipartite maximum cardinality $b$-matching. 
The maximum cardinality $b$-matching problem is defined in~\cref{def:b-matching}.

\begin{definition}[Maximum Cardinality Bipartite $b$-Matching (\mcbm)]\label{def:b-matching}
    Given an undirected, unweighted, bipartite graph $G = (L \cup R, E)$ 
    and a set of values $\{b_v \leq |R| \mid v \in L \cup R\}$, a 
    \defn{maximum cardinality $b$-matching} (\mcbm) finds a matching of maximum cardinality 
    between vertices in $L$ and $R$ where each vertex $v \in L \cup R$ is 
    matched to at most $b_v$ other vertices.
\end{definition}

The key difference between our algorithm for $b$-matching and the \mcm algorithm of~\cite{ALT21}
is that we have to account for when more than one item is assigned
to each bidder in $L$; in fact, up to $b_i$ items in $R$ can be assigned to any bidder $i \in L$. 
This one to many relationship calls for a different algorithm and analysis. The crux of our 
algorithm in this section is to create $b_i$ copies of each bidder $i$ and $b_j$ copies of each item $j$. 
Then, copies of items maintain their own prices and copies of bidders can each choose at most one item. We define some notation to
describe these copies. Let $\copies_i$ be the set of copies of bidder $i$ and $\copies_j$ be the set of copies of item $j$.
Then, we denote each copy of $i$ by $i^{(k)} \in C_i$ for $k \in [b_i]$ and each copy of $j$ by $j^{(k)} \in C_j$ for $k \in [b_j]$.
As before, we denote a bidder and their currently matched item by $\left(i^{(k)}, a_{i^{(k)}}\right)$. 

In \mcbm, we require that the set of all items 
chosen by different copies of the same bidder to include at most one copy of each item. 
In other words, we require if $j^{(k)} \in \bigcup_{i' \in C_i} a_{i'}$, then no other $j^{(l)} \in \bigcup_{i' \in C_i} a_{i'}$
for any $j^{(k)}, j^{(l)} \in C_j$ and $k \neq l$.
This \emph{almost} 
reduces to the problem of finding a maximum cardinality matching in a $\sum_{i \in L} b_i + \sum_{j \in R} b_j$ sized bipartite 
graph but \emph{not quite}. Specifically, the main challenge we must handle is when multiple copies of the 
same bidder want to be matched to copies of the \emph{same} item. In this case, we cannot match any of these bidder copies to
copies of the same item and thus must somehow handle the case when there exist items of lower price but we cannot match them.

In addition to handling the above hard case, 
as before, the crux of our proof relies on a variant of the $\eps$-happy definition 
and the definitions of appropriate potential functions.

Recall from the MCM algorithm
of~\cite{ALT21} that an $\eps$-happy bidder has utility that is at least the utility gained from matching to any other item (up to
an additive $\eps$). Such a definition is insufficient in our setting since it may be the case that matching to a copy 
of an item that is already matched to a different copy of the same bidder results in lower cost. However, such a match is not helpful
since any number of matches between copies of the same bidder and copies of the same item contributes a value of one to the cardinality of the 
eventual matching.

Our algorithm solves all of the above challenges and provides a $(1-\eps)$-approximate \mcbm in asymptotically the same number of 
rounds as the MCM algorithm of~\cite{ALT21}.
We describe our auction based algorithm for \mcbm next and the precise pseudocode is given in~\cref{alg:b-matching}. Our algorithm
uses the parameters defined in~\cref{table:mcbm-params}. We show the following results using our algorithm. We discuss semi-streaming implementations of our algorithm in~\cref{sec:b-matching-semi-streaming}. Let $L$ be the half with fewer numbers of nodes.

\mcbmthm*

\begin{algorithm}[htb]
    \caption{\label{alg:b-matching} Auction Algorithm for Bipartite $b$-Matching}
    \hspace*{\algorithmicindent} \textbf{Input:} Graph $G = (L \cup R, E)$ and parameter
    $0 < \eps < 1$. \\
    \hspace*{\algorithmicindent} \textbf{Output:} An $(1 - \eps)$-approximate 
    maximum cardinality bipartite $b$-matching.
    \begin{algorithmic}[1]
        \State Create $L', R', E'$ and graph $G'$. (Defined in~\cref{table:mcbm-params}.)\label{bm:copies}
        \State For each $i' \in L'$, set $(i', \bot)$ where $a_{i'} = \bot$, $c_{i'} \leftarrow 0$.\label{bm:assign-item}
        \State For each $j' \in R'$, set $p_{j'} = 0$.\label{bm:set-prices}
        \State Set $M_{\max} \leftarrow \emptyset$.\label{bm:max}
        \For{$d = 1, \dots, \bmphases$}\label{bm:phases}
            \State For each unmatched bidder $i' \in L'$,
            find $D_{i'} \leftarrow $ FindDemandSet$(G', i', c_{i'})$ [\cref{alg:find-demand-set}].\label{bm:demand}
            \State Create $G'_d$.\label{bm:subgraph}
            \State Find any arbitrary \nonduplicate maximal matching $\mmd$ of $G_d$.\label{bm:matching}
            \For{$(i', j') \in \mmd$}\label{bm:rematch-1}
                \State Set $a_{i'} = j'$ and $a_{i_{prev}} = \bot$ 
                for the previous owner $i_{prev}$ of $j'$.\label{bm:rematch-2}
                \State Increase $p_{j'} \leftarrow p_{j'} + \eps$.\label{bm:increase-price}
            \EndFor
            \State For each $i' \in L'$ with $D_{i'} \not= \emptyset$ and $a_{i'} = \bot$, increase $c_{i'} \leftarrow c_{i'} + \eps$.\label{bm:increase-price-cutoff}
            \State Using $M'_d$ compute $M_d$ where for each $(i', j') \in M'_d$, add
            $(i, j)$ to $M_d$ if $(i, j) \not\in M_d$.\label{bm:original}
            \State If $|M_d| > |M_{\max}|$, $M_{\max} \leftarrow M_d$.\label{bm:replace-max}
        \EndFor
    \State Return $M_{\max}$. \label{bm:return}
    \end{algorithmic}
\end{algorithm}

\begin{algorithm}
    \caption{\label{alg:find-demand-set} FindDemandSet$\left(G' = (L' \cup R', E'), i', c_{i'}\right)$.}
    \begin{algorithmic}[1]
        \State Let $N'(i') = \{j' \in R' \mid j^{(l)} \not= a_{i^{(k)}} \forall i^{(k)} \in C_{i}, \forall j^{(l)} \in C_j \wedge p_{j'} \geq c_{i'} \forall j' \in C_j\}$.\label{ds:n}
        \State $D_{i'} \leftarrow \arg\min_{j' \in 
            N'(i'), p_{j'} < 1} \left(p_{j'}\right)$.\label{ds:d}
        \State Return $D_{i'}$.
    \end{algorithmic}
\end{algorithm}

\subsection{Algorithm Description}
The algorithm works as follows. 
We assign to each bidder, $i$, $b_i$ unmatched slots and the goal is to fill all slots (or as many as possible). 
For each bidder $i \in L$ and each item $j \in R$, we create $b_i$ and $b_j$ copies, respectively, and assign these copies to 
new sets $L'$ and $R'$, respectively (\cref{bm:copies}). This step of the algorithm changes slightly in our streaming implementation.
For each bidder and item with an edge between them $(i, j) \in E$, we create a 
biclique between $C_i$ and $C_j$; the edges of all created bicliques is the set of edges $E'$. The graph $G' = (L' \cup R', E')$
is created as the graph consisting of nodes in $L' \cup R'$ and edges in $E'$.
As before, we initialize each bidder's
assigned item to $\bot$ (\cref{bm:assign-item}).
Then, we set the price for each copy in $R'$ to $0$ (\cref{bm:set-prices}). 

In our \mcbm algorithm, we additionally set a price cutoff for each bidder $c_{i'}$ initialized to $0$ (\cref{bm:assign-item}). Such a 
cutoff helps us to prevent bidding on lower price items previously not bid on because they were matched to another copy of the same bidder.
More details on how the cutoff prevents bidders from bidding against themselves can be found in the proof of~\cref{lem:unique-assign}.
We maintain the maximum cardinality matching we have seen in $M_{\max}$ (\cref{bm:max}).
We perform $\bmphases$ rounds of assigning items to bidders (\cref{bm:phases}). For each round, we first find the demand set
for each unmatched bidder $i' \in L'$ using~\cref{alg:find-demand-set}
(\cref{bm:demand}). The demand set is defined with respect to the cutoff price $c_{i'}$
and the set of items assigned to other copies of bidder $i$. The
demand set considers all items $j' \in R'$ that are neighbors of $i'$
where no copy of $j$, $j^{(k)} \in C_j$, is assigned to any copies of $i$
and $p_{j'} \geq c_{i'}$
(\cref{alg:find-demand-set}, \cref{ds:n}). From this set of neighbors,
the returned demand set is the set of item copies with the minimum price in $N'(i')$ (\cref{ds:d}).

Using the induced subgraph of 
$\left(\bigcup_{i'\in L'} D_{i'}\right) \cup L'$ (\cref{bm:subgraph}), we greedily find a maximal matching while avoiding
assigning copies of the same item to copies of the same bidder (\cref{bm:matching}). We call such a maximal matching 
that does not assign more than one copy of the same item to copies of the same bidder to be a
\defn{\nonduplicate} maximal matching. This greedy matching prioritizes
the unmatched items by first matching the unmatched items and then matching the matched items. We can
perform a greedy matching by matching an edge if the item is unmatched and 
no copies of the bidder it will match to is matched to another copy of the item.
For each newly matched item (\cref{bm:rematch-1}), we rematch the item to the 
newly matched bidder (\cref{bm:rematch-2}). We increase the price
of the newly matched item (\cref{bm:increase-price}). 
For each remaining unmatched bidder, we increase the cutoff price by $\eps$ (\cref{bm:increase-price-cutoff}).

We compute the corresponding matching in the original graph using $M'_d$
(\cref{bm:original}) by including one edge $(i, j)$ in the matching 
if and only if there exists at least one bidder copy $i' \in C_i$ 
matched to at least one copy of the item $j' \in C_{j}$.
Finally, we return the maximum cardinality $M_{\max}$
matching from all iterations as our
$(1-\eps)$-approximate maximum cardinality $b$-matching 
(\cref{bm:return}).

\subsection{Analysis}

In this section, we analyze the approximation error of our algorithm and prove that it provides a $(1-\eps)$-approximate maximum cardinality $b$-matching.

\paragraph{Approach}
We first provide an intuitive explanation of the approach we take to perform our analysis and then we give our precise analysis. Here, we describe both the 
challenges in performing the analysis and explain our choice of certain methods in the algorithm to facilitate our analysis. We especially highlight the parts
of our algorithm and analysis that differ from the original \mcm algorithm of~\cite{ALT21}. First, in order to show the approximation factor of our algorithm, 
we require that the utility obtained by a large number of matched bidders from our algorithm is greater than the corresponding utility from switching to the optimum 
items in the optimum matching. For $b$-matching, any combination of matched items and bidder copies satisfy this criteria. Furthermore, matching multiple
item copies of the same item to bidder copies of the same bidder does not increase the utility of the bidder. Thus, we look at matchings where
at most one copy of each bidder is matched to at most one copy of each item. Recall our definition of $\eps$-happy given in~\cref{def:happy}
and we let $\happy$ be the set of bidders satisfying that definition.

For $b$-matching, each bidder $i$ is matched to a set of at most $b_i$ items. Let $(i, O_i) \in \opt$ denote the set of items 
$O_i \subseteq R$ matched to bidder $i$ in $\opt$.
Recall 
\ifcameraready
\else
from~\cref{sec:mcm} 
\fi
that the proof requires $u_i \geq 1 - p_{o_i} - \eps$ for 
every bidder $i \in \happy \cap \opt$ to show that $\sum_{i \in L} u_i \geq \sum_{i \in \happy \cap \opt} 1 - p_{o_i} - \eps$. 
Using our bidder copies, $C_i$,
the crux of our analysis proof is to show that for every $(i, O_i) \in \opt$, we can \emph{assign} the 
items in $O_i$ to the set of happy bidder copies in $C_i$ such that each happy bidder copy receives a unique item, denoted
by $r_{i'}$, and $c_{i'} \leq p_{\min, r_{i'}}$ where $p_{\min, r_{i'}}$ is the price of the minimum priced
copy of $r_{i'}$. Using this assignment, we are able to show once again that $\sum_{i' \in L'} u_{i'} \geq \sum_{i' \in \happy \cap \opt}
1 - p_{\min, r_{i'}} - \eps$. This requires a precise definition of $\happy \cap \opt$. Let $S_i \subseteq C_i$ be the set of all happy bidders in $C_i$.
Recall that the optimum solution gives a matching between a bidder $i \in L$ and potentially multiple items in $R$; we 
turn this matching into an optimum matching in $G'$. 
If $|S_i| \leq |O_i|$, then all 
happy copies in $S_i$ are in $\opt$; otherwise, we pick an arbitrary set of $|O_i|$ happy bidder copies in $S_i$ to be in
$\opt$. Then, the summation is determined based on this set of happy bidder copies in $\happy \cap \opt$.

Once we have shown this, the only other remaining part of the proof is to show that in the $\ceil{\frac{2}{\eps^2}}$ rounds that
we run the algorithm the potential increases by $\eps$ for every unhappy bidder in $\opt$ for each round that the bidder is unhappy.
As in the case for \mcm, the price of an item increases whenever it becomes re-matched. Hence, $\Pi_{items}$ increases by $\eps$
each time a bidder who was happy becomes unhappy. To ensure that $\Pi_{bidders}$ increases by $\eps$ for each bidder who was unhappy
and remains unhappy, we set a \defn{cutoff} price that increases by $\eps$ for each round where a bidder remains unhappy. Thus, this
cutoff guarantees that $\Pi_{bidders}$ increases by $\eps$ each time. 

\paragraph{Detailed Analysis} Now we show our detailed analysis that formalizes our approach described above.
We first show that our algorithm maintains both~\cref{inv:nonzero-matched-1} and~\cref{inv:max-utility-1}. 
We also show our algorithm obeys the following invariant.

\begin{invariant}\label{inv:one-item-match}
    The set of matched items of all copies of any bidder $i \in L$ contains at most one copy of each item. In other words, $\left| \bigcup_{i' \in C_i} a_{i'} \cap C_j \right| \leq 1$ for
    all $j \in R$.
\end{invariant}

We restate two invariants used in~\cite{ALT21} below. We prove that our~\cref{alg:b-matching} also maintains these two invariants. 

\begin{invariant}[Non-Zero Price Matched~\cite{ALT21}]\label{inv:nonzero-matched-1}
    Any item $j$ with positive price $p_j > 0$ is matched.
\end{invariant}

\begin{invariant}[Maximum Utility~\cite{ALT21}]\label{inv:max-utility-1}
 The total utility of all bidders is at most the cardinality of the matching minus the total price of the items.
\end{invariant}

\begin{lemma}\label{lem:b-inv}
    \cref{alg:b-matching} maintains~\cref{inv:one-item-match},~\cref{inv:nonzero-matched-1}, and~\cref{inv:max-utility-1}. 
\end{lemma}

\ifcameraready
\else
\begin{proof}
    An item increases in price only when it is matched to a bidder by~\cref{bm:increase-price}. A matched item never becomes unmatched in our algorithm. Thus,
    \cref{inv:nonzero-matched} is maintained.
    By definition of utility, the utility obtained from the matching produced by~\cref{alg:b-matching} is $\sum_{i' \in L'} u_{i'} = 
    \sum_{i' \in L'} 1 - p_{a_{i'}} \leq |M| - \sum_{i' \in L'} p_{a_{i'}}$. Hence,~\cref{inv:max-utility} is also satisfied by our algorithm.

    Suppose for contradiction that~\cref{inv:one-item-match} is violated at some point in our algorithm. Then, suppose $i^{(k)}, i^{(l)} \in C_i$ are two copies of
    bidder $i$ that are matched to two copies of the same item. Either they matched to two copies of the same item in the same round
    or they matched to the items in different rounds. In the first case,~\cref{bm:matching} ensures no two copies of the same bidder
    are matched to copies of the same item in the same round.
    In the second case, suppose without loss of generality that $i^{(l)}$ was matched after $i^{(k)}$. Then, this means that $D_{i^{(l)}}$ contains a
    copy of of the same item that is matched to $i^{(k)}$. This contradictions how $D_{i^{(l)}}$ was constructed in~\cref{ds:n}. Thus, \cref{inv:one-item-match} 
    follows. 
\end{proof}
\fi

We follow the style of analysis outlined 
\ifcameraready
\else
in~\cref{sec:mcm} 
\fi
by defining appropriate definitions of $\eps$-happy and appropriate
potential functions $\Pi_{items}$ and $\Pi_{bidders}$.
In the case of $b$-matching, we modify the definition of $\eps$-happy
in this setting to be the following.

\begin{definition}[$(\eps, c)$-Happy]\label{def:b-eps-happy}
    A bidder $i' \in L'$ is $(\eps, c)$-happy (at the end of a round)
    if $u_{i'} \geq 1 - p_{j'} - \eps$ for 
    all neighbors in the set $N'(i')$ where $N'(i')$ is as defined in~\cref{ds:n} of~\cref{alg:find-demand-set} 
    (i.e.\ contains all neighboring
    items $j'$ where $p_{j'} \geq c_{i'}$ and no copy of the neighbor is matched to another copy of $i'$).
\end{definition}

At the end of each round, it is easy to show that all matched $i'$ and $i'$
whose demand sets $D_{i'}$ are empty are $(\eps, c_{i'})$-happy.

\begin{lemma}\label{lem:b-eps-happy}
    At the end of any round, if bidder $i'$ is matched or if their demand
    set is empty, $D_{i'} = \emptyset$, then $i'$ is $(\eps, 
    c_{i'})$-happy.
\end{lemma}

\ifcameraready
\else
\begin{proof} %
First, consider the case when the demand set $D_{i'}$ is empty. 
Let $c_{i'}$ be the cutoff price at the \emph{end} of the round. This means that $i'$ remains unmatched at the 
end of the round and $c_{i'}$ \emph{does not} increase from the beginning of the round since $D_{i'}$ is empty.
In this case, it means that all neighboring items with
price $\geq c_{i'} - \eps$ and which were not matched to another copy of $i$ 
at the beginning of the round had price $1$. Then, the utility that 
can be gained from any of these items is $0$ and our bidder $i'$, who has utility $u_{i'} = 0$, 
is $(\eps, c_{i'})$-happy.

Suppose that instead $i'$ is matched. Then, $i'$ must have matched to an item from its demand 
set. Recall that the demand set consists of the lower priced items from the set of $i'$s neighbors
with price at least $c_{i'}$ and which were not matched to any copy of $i$. This is precisely the set of neighbors
we are comparing against. Since we matched against one of the lowest priced items in this set and the price of the
item increases by $\eps$ after being matched, the utility is lower bounded by $1 - p_{j'} - \eps$ for all
$j' \in N'(i')$.
\end{proof}
\fi

In addition to the new definition of happy, we require another crucial observation before we prove our
approximation guarantee. Specifically, we show that for any set of bidder 
copies $C_i$ and any set of $|C_i|$ items $I \subseteq R$,
\cref{lem:b-eps-happy} is sufficient to imply there exists at least one assignment of items in $I$
to happy bidders in $S_i$ such that each item is assigned to at most one bidder and each happy bidder is assigned at 
least one item where the minimum price of the item is at least the cutoff price of the bidder.

\begin{lemma}\label{lem:unique-assign}
    For a set of bidder copies $C_i$ and
    any set $I \subseteq R$ of $|C_i|$ items where $(i, j) \in E$ for all items $j \in I$, 
    there exists at least one assignment of 
    items in $I$ to bidders in $C_i$, where we denote the item assigned to copy $i'$ by $r_{i'}$, that satisfy the following
    conditions:
    
    \begin{enumerate}
        \item The assignment is a one-to-one mapping between bidders in $C_i$ and items in $I$.
        \item Any item $j$ matched to $i'$ is assigned to $i'$.
        \item Let $r^{*}_{i'}$ be the lowest cost copy of item $r_{i'}$, $r^{*}_{i'} = \arg\min_{j' \in C_{r_{i'}}} \left(p_{j'}\right)$; then
        $p_{r^{*}_{i'}} \geq c_{i'}$ for all $i' \in C_i$.
    \end{enumerate} 
\end{lemma}

\ifcameraready
\else
\begin{proof} %
In this proof, we prove a stronger statement which is sufficient to prove our original lemma statement. Namely, we prove that 
for each bidder $i' \in L'$, during any round $d \leq \ceil{\frac{2}{\eps^2}}$, of the items in $N(i)$, 
at most $|C_i| - 1$ of them can have minimum price $< c_{i'}$ and each of these items can be 
assigned to a unique copy of $C_i$ that is not $i'$. This means that any subset of $|C_i|$ items in $N(i)$
containing the items with minimum
price $< c_{i'}$ can be assigned to these unique copies and rest of the items can be arbitrarily assigned to any of the remaining
copies of $C_i$.

We now prove the above. Let $i' \in L'$ be any bidder in $L'$. We say an item's minimum priced copy \emph{falls below} $c_{i'}$ when $c_{i'}$
increases above the minimum priced copy of an item.
An item's price falls below $c_{i'}$ only when another copy of the item is matched to another copy of $i$. 
Now we first argue there cannot be more than $|C_i| - 1$ of these items. To show this, we first show that each copy $i'' \in C_i$  
where $i' \neq i''$ can cause at most one item in $R$ to have a copy with minimum price less than $c_{i'}$. We say a bidder copy
$i''$ \emph{caused} $j$ to have minimum price less than $c_{i'}$ if $i''$ was matched to a copy of $j$ in the earliest round when 
the minimum priced copy of $j$ drops below $c_{i'}$ and does not have price $\geq c_{i'}$ in any later rounds up to the current round.
In other words, suppose the current round is $d$ and the minimum priced copy of $j$ dropped below $c_{i'}$ in round $d' < d$ because
it was matched to item $i''$. Then suppose the minimum priced copy of $j$ does not exceed $c_{i'}$ again after round $d'$. We
say that $i''$ \emph{caused} $j$ to have minimum price less than $c_{i'}$.

Suppose for contradiction that $i''$ can cause more than one item to have minimum price less than $c_{i'}$. 
Then, suppose $i''$ caused both $j_1, j_2 \in R$ where $j_1 \neq j_2$ to have a copy
with minimum price smaller than $c_{i'}$. Without loss of generality, assume bidder $i''$ was initially matched to a copy of $j_1$
and then to a copy of $j_2$. There are again several cases to consider.

Bidder $i''$ may have switched to a copy of $j_2$ from a copy of $j_1$ during some round when the minimum priced copies of both items 
were the same. If they have price equal to $c_{i'}$, then, they can have minimum price $< c_{i'}$ in the subsequent round if and only if 
both are matched to copies of $i'$. In that case, $i''$ cannot cause $j_2$ to have minimum price less than $c_{i'}$. 
Suppose both item's minimum prices are less than $c_{i'}$. Then, at some point $j_2$ must have been matched to some copy of $i'$ to drop 
below $c_{i'}$ in price. Without loss of generality, suppose this is the first time that $i''$ switched its matching to $j_2$ since the minimum 
priced copy of $j_2$ dropped below $c_{i'}$. Then, $i''$
cannot have caused the minimum price of $j_2$ to drop below $c_{i'}$ since $j_2$ already has minimum price below $c_{i'}$ when
$i''$ switched to it. Since each item which falls below $c_{i'}$ requires a unique copy in $C_i$ (which is not $i'$) 
there can be at most $|C_i| - 1$ such items.

Now, we conclude the proof by showing each such item with minimum price less than $c_{i'}$ can be assigned to a unique copy of $C_i$.
We proved above that a unique copy of $i$ caused each item to drop below $c_{i'}$ in price. Furthermore, we also proved above that 
a bidder can switch to another item if and only if the items have the same minimum price. A bidder $i_1$ can be assigned to the item $j_1$
they originally caused to drop below $c_{i'}$ in price unless $j_1$'s price drops below $c_{i_1}$. Suppose without 
loss of generality that this is the first such bidder whose original item fell below its cutoff price. 
Then, there must exist another bidder
$i_2 \in C_i$ who matched to $j_1$ and was assigned item $j_2$
that has the same minimum price as $j_1$. We switch the assignments of $j_2$ to $i_1$ and $j_1$ to $i_2$ in this case. 
We perform this switch sequentially for every such bidder whose original item fell below its cutoff price.
Thus, we showed that each bidder in $C_i$ can either be assigned to the item they originally caused to drop below $c_{i'}$ or we can
switch the assignment of two such bidders.
\end{proof} 
\fi

We now perform the approximation analysis. Suppose as in the case of MCM, we have at least 
$(1-\eps)|\opt|$ happy bidders in $\opt$ (i.e. $|\happy \cap \opt| \geq (1-\eps)|\opt|)$, 
then we show that we can obtain a $(1-\eps)$-approximate \mcbm. 
Let $\opt$ be an optimum \mcbm matching and $|\opt|$ be the cardinality of this matching.

\begin{lemma}\label{lem:mcbm-approx}
    Assuming $|\happy \cap \opt| \geq (1-\eps)|\opt|$, then we obtain
    a $(1-2\eps)$-approximate \mcbm.
\end{lemma}

\ifcameraready
\else
\begin{proof} %
Let $\opt$ be an optimum \mcbm and $(i, J) \in \opt$ be the bidder and item set pairs in $\opt$. 
Let $|\opt|$ be the cardinality of the optimum matching.
Using~\cref{lem:unique-assign}, for each pair $(i, O_i)$, we \emph{assign} the items in $O_i$
to $C_i$. Now, we upper and lower bound the utility of all 
matched bidders as before using this assignment. The upper bound is the same as the case for MCM.
\begin{align*}
    |M| - \sum_{j \in R'} p_j &\geq \sum_{i \in L'} u_i
\end{align*}
since all items with non-zero price is assigned to a bidder and the maximum cardinality cannot exceed
the cardinality of the obtained matching $M$.

Then, to lower bound the sum of the utilities we obtain for each pair of bidder copy and assigned item
\begin{align*}
    u_{i'} = 1 - \eps - p_{o_{i'}}
\end{align*}
by~\cref{lem:unique-assign} where $o_{i'} \in O_i$ is the item assigned to $i$ and where $p_{o_{i'}} 
= \arg\min_{j' \in C_{o_{i'}}} \left(p_{j'}\right)$. 
By~\cref{lem:b-eps-happy}, the above equation follows.

This means that summing over all happy bidders results in

\begin{align*}
    \sum_{i' \in L} u_{i^{'}} &\geq \sum_{i' \in \happy \cap \opt} 1 - \eps - p_{o_{i'}}\\
    &\geq (1-\eps)|\opt| - \sum_{i' \in \happy \cap \opt} (\eps - p_{o_{i'}})\\
    &\geq (1-2\eps)|\opt| - \sum_{i' \in \happy \cap \opt} p_{o_{i'}}
\end{align*}

Combining the lower and upper bounds we obtain our desired approximation ratio

\begin{align*}
    |M| - \sum_{j' \in R'} p_{j'} &\geq (1-2\eps)|\opt| - \sum_{i' \in \happy \cap \opt} p_{o_{i'}}\\
    |M| &\geq (1-2\eps)|\opt|
\end{align*}
\end{proof}
\fi

The potential argument proof is almost identical to that for MCM provided our use of $c_{i'}$.
Specifically, as in the case for MCM, we use the same potential functions and using these potential functions,
we show that our algorithm terminates in $O\left(\frac{1}{\eps^2}\right)$ rounds. The key difference between our
proof and the proof of \mcm explained 
\ifcameraready
\else
in~\cref{sec:mcm} 
\fi
is our definition of $\Pi_{bidders}$ which is precisely defined
in the proof of~\cref{lem:mcbm-rounds} below.

\begin{lemma}\label{lem:mcbm-rounds}
    In $\bmphases$ rounds, there exists at least one round where $|\opt \cap \happy| \geq (1-\eps)|\opt|$. 
\end{lemma}

\ifcameraready
\else
\begin{proof}%
We use similar potential functions as used in~\cite{ALT21} (\cref{sec:mcm}) with the difference being the definition of $\bidderpot$.
We define $\bidderpot$ by picking an arbitrary set of $|O_i|$ bidder copies for each $i \in L'$ to be contained in the set $\opt$. We let this 
set of copies be denoted as $\opt$. Then, we 
define the potential functions as follows:

\begin{align*}
    \Pi_{items} &\defined \sum_{j' \in R'} p_{j'}\\
    \Pi_{bidders} &\defined \sum_{i' \in \opt} \min_{j' \in N'(i'), p_{j'} < 1} \left(p_{j'}\right).
\end{align*}

First, both $\Pi_{items}$ and $\Pi_{bidders}$ are upper bounded by $|\opt|$ since the price of any item
is at most $1$ and the number of non-zero priced items is precisely the number of matched items by~\cref{inv:nonzero-matched}. 
We show that having at least $\eps \cdot |\opt|$ bidders in $\opt$ that are not happy
increases the potential on one or both of the potential functions by at least $\eps^2 \cdot |\opt|$. 

When a bidder becomes unmatched, the price of its previously matched item increases by $\eps$. When a 
bidder remains unmatched, its $\min_{j' \in N'(i'), p_{j'} < 1} \left(p_{j'}\right)$ increases by $\eps$, by~\cref{bm:increase-price-cutoff}.
Thus, in all settings, for each unhappy bidder, either $\Pi_{items}$ increases by $\eps$ or $\Pi_{bidders}$ increases
by $\eps$. The total potential for both is $2 \cdot |\opt|$ and so we obtain $\frac{2\cdot |\opt|}{\eps^2 \cdot |\opt|} \leq \bmphases$ rounds.

By our definition of $\happy \cap \opt$, if $\geq (1-\eps)|\opt|$ are happy, then $|\happy \cap \opt| \geq (1-\eps)|\opt|$.
\end{proof}
\fi

Using the above lemmas, we can prove the round complexity of~\cref{thm:b-matching} to be $O\left(\frac{1}{\eps^2}\right)$ by~\cref{lem:mcbm-approx}
and~\cref{lem:mcbm-rounds}.

\begin{theorem}\label{lem:mcbm-distributed}
    There exists an auction algorithm for maximum cardinality
    bipartite b-matching (\mcbm) that gives a $(1-\eps)$-approximation for any $\eps > 0$ and runs in 
    $O\left(\frac{\log n}{\eps^2}\right)$ rounds of communication
    using $O(b \log n)$ bits per message in the blackboard distributed model.
    In total, the number of bits used
    by the algorithm is $O\left(\frac{nb \log^2 n}{\eps^2}\right)$.
\end{theorem}

\subsection{Semi-Streaming Implementation}\label{sec:b-matching-semi-streaming}

We now show an implementation of our algorithm to the semi-streaming setting and show the following lemma which proves the semi-streaming
portion of our result in~\cref{thm:b-matching}.  We are guaranteed $\eps \geq \frac{1}{2n^2}$; otherwise, an exact matching is found.
In order to show the space bounds, we use an additional lemma below that upper and lower bounds the prices of any copies of the same
item in $R'$.

\begin{lemma}\label{lem:price-difference}
    For any $j \in R$, let $j_{\min}$ be the minimum priced copy in $C_j$ and $j_{\max}$ be the maximum priced copy in $C_j$. Then, $p_{j_{\max}} - 
    p_{j_{\min}} \leq \eps$.
\end{lemma}

\ifcameraready
\else
\begin{proof}
We prove this lemma via contradiction. Suppose for contradiction that $p_{j_{\max}} - p_{j_{\min}} > \eps$ for some $j \in R$. This means that
during some round $d$, a bidder $i' \in L'$ matched to an item copy $j'$ where $p_{j'} > p_{j_{\min}}$. By~\cref{alg:find-demand-set}, this
can only happen if $D_{i'}$ contains $j'$ but not $j_{\min}$. If $j' \in D_{i'}$, then by definition of $N'(i')$, it holds 
that $j_{\min} \geq c_{i'}$ and 
no copy of $j$ is matched to another copy of $i$. Then, $j_{\min} \in N'(i')$ and $j_{\min} \in \arg\min_{j' \in N'(i')} \left(p_{j'}\right)$, 
a contradiction to $j' \in D_{i'}$ since $p_{j'} > p_{j_{\min}}$. 
\end{proof}
\fi

Using the above, we prove our desired bounds on the number of passes and the space used.

\begin{theorem}\label{lem:b-semi-streaming}
    There exists a semi-streaming algorithm for maximum cardinality bipartite $b$-matching that uses $O\left(\frac{1}{\eps^2}\right)$ rounds and 
    $\tO\left(\left(\sum_{i \in L} b_i + |R|\right) \log(1/\eps)\right)$ space where $L$ is the side with the smaller number of nodes in the input graph. 
\end{theorem}

\ifcameraready
\else
\begin{proof} %
We implement the steps in~\cref{alg:b-matching} in the semi-streaming model and show that they can be implemented within the bounds of this lemma.
We maintain in memory the following:

\begin{enumerate}
    \item The tuples $(i', a_{i'})$ for each $i' \in L'$, and
    \item The minimum and maximum prices for each item $j \in R$ and a count of the number of item copies at the minimum price and the maximum price for each item.
\end{enumerate}

For each round (\cref{bm:phases}), we spend one pass finding the minimum price of items in the $N'(i')$ of each bidder $i' \in L'$.
Then we spend another pass greedily finding a non-duplicate maximal matching among the items that have this minimum price.
To find a non-duplicate maximal matching that prioritizes unmatched items, we perform two passes in our streaming algorithm.
During the first pass, for each edge we receive in the stream, we first check that the minimum price of the item equals the 
demand set price. If this condition is satisfied and the following are also true, 
\begin{enumerate}
    \item at least one copy of the bidder adjacent to the edge is unmatched and has sufficiently low cutoff price,
    \item none of the copies of the bidder matched to any copies of the item,
    \item and at least one minimum priced copy of the item is unmatched,
\end{enumerate} 
then we match an unmatched copy of the item 
with an unmatched copy of the bidder (with sufficiently low cutoff price). We can do this greedily in the streaming setting since 
we maintain all copies of bidders in memory as well as the minimum 
and maximum prices of all items. This means that we can check 
all copies of all bidders to find an unmatched copy. Furthermore, 
we maintain pointers from items to their matched bidder copies so 
we can check the pointers as well as the minimum prices of items and
their counters to greedily find the appropriate matchings.

In the second pass, we match the matched items in the same manner as before in the first pass, except we consider all items in
each node's demand set (not just unmatched ones).
Reallocating the items and increasing the prices of rematched items can be done from the matching above in 
$\tO\left(\left(\sum_{i \in L} b_i + |R|\right) \log(1/\eps)\right)$ space without needing additional passes from the stream.
Finally, computing $M_d$ can also be done using $M'_d$ in the same amount of memory without additional passes of the stream.
\end{proof}
\fi

We note that the space bound is necessary in order to report the solution. (There exists a given input where reporting the solution requires 
$\tO\left(\left(\sum_{i \in L} b_i + |R|\right) \log(1/\eps)\right)$ space.) Thus, our algorithm is tight with respect to this notion. 

\subsection{Shared-Memory Parallel Implementation}

We now show an implementation of our algorithm to the 
shared-memory parallel setting. The main challenge for
this setting is obtaining an algorithm for obtaining
non-duplicate maximal matchings. To obtain 
non-duplicate maximal matchings, we just need to 
modify the maximal matching algorithm of~\cite{BFS12}
to obtain a maximal matching with the non-duplicate
characteristic. Namely, the modification we make is to
consider all copies of a node to be neighbors of each 
other. Since there can be at most $n$ copies of a node,
this increases the degree of each node by at most $n$.
Hence, the same analysis as the original algorithm 
still holds in this new setting.

\begin{theorem}\label{lem:b-matching-parallel}
    There exists a shared-memory parallel algorithm for 
    maximum cardinality bipartite $b$-matching that 
    uses $O\left(\frac{\log^3 n}{\eps^2}\right)$ depth 
    and $O\left(\frac{m \log n}{\eps^2}\right)$ total 
    work where $L$ is the side with the smaller number
    of nodes in the input graph. 
\end{theorem}

\ifcameraready
\else
\begin{proof}
    Finding the demand sets can be done using a parallel scan and sort in $O(m \log n)$ work 
    and $O(\log n)$ depth. Then, finding the induced subgraph can be done using a parallel 
    scan in $O(m)$ work and $O(\log n)$ depth. Finally, we use a modified version of the 
    maximal matching algorithm of~\cite{BFS12} to compute the maximal matching in each phase. 
    Our modified version of the algorithm of~\cite{BFS12} considers all copies of the same node
    to be neighbors of each other; all other parts of the algorithm remains the same. This means that
    the degree of each node increases by at most $n$ (resulting in a maximum degree of at most $2n$)
    which means that the asymptotic work and depth remains the same as before with $O(m)$ work and
    $O(\log^2 n)$ depth. Combined, we obtain the work and depth as stated in the lemma.
\end{proof}
\fi
\newpage
\appendix 
\onecolumn

\section{An Auction Algorithm for Maximum Cardinality Bipartite Matching~\cite{ALT21}}\label{sec:mcm}
This paper focuses on auction-based algorithms for various maximum matching problems. 
Traditionally, the exact versions of the 
maximum cardinality bipartite matching (MCM), the maximum weight bipartite
matching (MWM), and the maximum cardinality bipartite $b$-matching (\mcbm) %
problems have been solved using maximum flow or the Hungarian method.
The starting point for this paper is the auction-based algorithm of Assadi, Liu, and Tarjan~\cite{ALT21}. 
We first give a brief overview of their algorithm as well as the framework for their analysis. 
We then show extensions of their framework into the more general domains of bipartite $b$-matching (\mcbm) and
maximum weight bipartite matching (\mwm).

\paragraph{Auction-Based \mcm Algorithm (\cite{ALT21})}
The auction-based algorithm of Assadi, Liu, and Tarjan works as follows. Given a bipartite input graph $G = (L \cup R, E)$,
the \defn{bidders} in $L$ bid on the \defn{items} in $R$ in $\ceil{\frac{2}{\eps^2}}$ \defn{rounds} of bidding. 
Initially, items $j \in R$ are given prices of $p_j \leftarrow 0$. In each round, all bidders who are not matched to items 
compute a \defn{demand set}. The demand set $D_i$ of a bidder $i \in L$ consists of the lowest price neighbors of $i$ whose 
prices are less than $1$. In other words, $D_i \defined \arg\min_{j \in N(i), p_j < 1} \left(p_j\right)$. After all 
unmatched bidders determine their demand set, they create an induced subgraph consisting of all unmatched bidders and their
demand sets. In this induced subgraph, they find an arbitrary maximal matching $M$. Using this maximal matching, the items are
re-matched to new bidders. Suppose $(i, j) \in M$ is an edge in the maximal matching and $(i, a_i)$ is the tuple representing
the bidder $i$ and its matched item $a_i$. If $a_i = \bot$, then $i$ is unmatched. For each $(i, j) \in M$, they set $a_i = j$
and $a_{i'} = \bot$ where $i'$ is the previous bidder which was matched to $j$. 
Then, the price for $j$ increases
by $\eps$ as in $p_j \leftarrow p_j + \eps$. This entire process repeats for $\ceil{\frac{2}{\eps^2}}$ rounds and the 
resulting maximum matching out of all rounds is returned.

\paragraph{Analysis} The analysis of their algorithm consists of two key components: a notion of \defn{happy} bidders and potential
functions for unhappy bidders. Happy bidders are those whose \defn{utility} does not increase by more than $\eps$ if they were to be matched
to a different item. \defn{Unhappy} bidders, on the other hand, are those whose utility can increase by more than $\eps$ if 
they were matched to a different item. 
Such a notion is important when comparing the matching obtained by the auction-based algorithm against the 
optimum MCM. The \defn{utility} of a bidder $i$ is defined to be $u_i = 1 - p_{a_i}$ if $a_i \neq \bot$. Otherwise, if $a_i = \bot$, then
the utility of $i$ is $0$. Specifically, the notion of $\eps$-happy is defined to be the following:

\begin{definition}[$\eps$-Happy~\cite{ALT21}]\label{def:eps-happy}
    A bidder $i$ is $\eps$-happy if $u_i \geq 1 - p_j - \eps$ for every $j \in R$. 
\end{definition}

They show that if at least $(1-\eps)|\opt|$
of the bidders in $\opt$ (where $\opt$ is the maximum cardinality matching and $|\opt|$ is the cardinality of this matching)
are $\eps$-happy then their obtained matching is a $(1-2\eps)$-approximate \mcm. Intuitively, this is due to two facts. First, 
the following invariant is maintained.

\begin{invariant}[Non-Zero Price Matched~\cite{ALT21}]\label{inv:nonzero-matched}
    Any item $j$ with positive price $p_j > 0$ is matched.
\end{invariant}

Second, the next invariant is also maintained.

\begin{invariant}[Maximum Utility~\cite{ALT21}]\label{inv:max-utility}
 The total utility of all bidders is at most the cardinality of the matching minus the total price of the items.
\end{invariant}

These two invariants allow them to show, via the following calculation,
the desired approximation factor, assuming at least $(1-\eps)|\opt|$ of the bidders in $\opt$ are happy:

\begin{align}
    |M|  - \sum_{j \in R} p_j &\geq \sum_{i \in L} u_i \geq \sum_{i \in \opt \cap \happy} 1 - p_{o_i} - \eps \label{eq:happy-1}\\
    |M| - \sum_{j \in R} p_j &\geq (1-\eps)|\opt| - \sum_{i \in \opt \cap \happy} p_{o_i} - \sum_{i \in \opt \cap \happy} \eps \label{eq:happy-2}\\
    |M| - \sum_{j \in R} p_j &\geq (1-\eps)|\opt| - \eps|\opt| - \sum_{i \in \opt \cap \happy} p_{o_i} \label{eq:happy-3}\\
    |M| &\geq (1-2\eps)|\opt| \label{eq:happy-4}.
\end{align}

In the above equations, $\happy$ is the set of happy bidders in $L$ and $o_i$ is the item matched to bidder $i$ in $\opt$.
\cref{eq:happy-1} follows from~\cref{inv:max-utility} and the definition of happy (\cref{def:eps-happy}).
\cref{eq:happy-2} simplifies $\sum_{i \in \opt \cap \happy} 1 \geq (1-\eps)|\opt|$ by the assumption. 
\cref{eq:happy-3} follows since $\sum_{i \in \opt \cap \happy} \eps \leq \eps |\opt|$.
Finally, they obtain~\cref{eq:happy-4} using~\cref{inv:nonzero-matched} which implies that $\sum_{j \in R} p_j \geq \sum_{i \in \opt \cap \happy} p_{o_i}$.

Now, the only thing that remains to be shown is that in $\ceil{\frac{2}{\eps^2}}$ total rounds, there exists at least one round 
where $\geq (1-\eps)|\opt|$ of the bidders in $\opt$ are $\eps$-happy. They argue this through a clean and simple potential function argument.
They define two potential functions (below) that ensure that for each unhappy bidder $i$ that is also in $\opt$, the potential of 
one of these potential functions increases by $\eps$ for each round the bidder is unhappy:

\begin{align}
    \Pi_{items} &\defined \sum_{j \in R} p_j\label{pot:item}\\
    \Pi_{bidders} &\defined \sum_{i \in \opt} \min_{j \in N(i), p_j < 1} \left(p_{i}\right).\label{pot:bidders}
\end{align}

Both of the potential functions above are upper bounded by $|\opt|$. Otherwise, a higher potential implies a solution with larger
cardinality than $\opt$, a contradiction to the optimality of $\opt$. Thus, since each unhappy bidder increases the potential of at least
one of these potential functions by $\eps$, the total increase in potential when at least $\eps|\opt|$ of the bidders in $\opt$ are unhappy
is at least $\eps \cdot \eps|\opt|$.
Then, the total number of rounds necessary before they obtain at least one round where at least
$(1-\eps)|\opt|$ of bidders in $\opt$ are happy is upper bounded by $\ceil{\frac{2|\opt|}{\eps \cdot \eps|\opt|}} = \ceil{\frac{2}{\eps^2}}$.

\section{Gupta-Peng~\cite{GP13} Transformation}\label{app:guptapeng}

We state modified versions of the Gupta-Peng~\cite{GP13} transformation in this section that can be applied
to the distributed, parallel, and streaming settings. Our transformations are almost identical to the analysis 
given by~\cite{GP13} and we encourage interested readers to refer to the original work for the 
original analyses and to~\cite{BDL21} for adaptations to some of the different settings. For completeness and to
make our paper self-contained, we include all relevant proofs in this paper. 
The purpose of the transformation is to take an algorithm which obtains an $(1-\eps)$-approximate maximum weighted
matching with a complexity measure that has a polynomial dependency on the maximum weight in the input graph and 
convert it into an algorithm with some greater dependency on the approximation parameter $\eps > 0$ and polylogarithmic 
dependency on the maximum weight in the graph.
The transformation works by maintaining several versions of a blackbox $(1-\eps)$-approximate maximum weighted
matching algorithm on smaller instances of the problem to obtain a $(1-\eps)$-approximate maximum weighted matching
algorithm with the desired new complexity bounds. 

For the remainder of this section, to be consistent
with the notation used in~\cite{GP13}, we refer
to the approximations as ``$(1+\eps)$-approximations''. 
Such approximations can be easily converted to 
$(1-\eps)$-approximations used as our notation for the 
rest of this paper.
The transformation proceeds as follows. We first define some notation used to describe the algorithm. Let an edge $e = (u, v)$
be in level $\ell$ if its weight is in a certain
range to be determined later.
Then, let $\hatM_{\ell}$ be a matching found for level $\ell$
by a $(1+\eps)$-approximate maximum weighted matching
algorithm. Then, the approximate matching for the entire 
graph is produced by iterating from the largest
$\ell$ to the smallest $\ell$ and greedily choose
edges in $\hatM_{\ell}$ to add to the matching $\hat{M}$
as long as the chosen edge is not adjacent to any 
endpoint of an edge in $\hat{M}$.
Let $\mathcal{R}(e)$ for an edge $e = (u, v)$ be defined as $\mathcal{R}(e) = \{e\} \cup \{(x, y) \mid (x, y) \in \hatM_{\ell'} \text{ where }
\ell' < \ell, \text{ and } \{x, y\} \cap \{u, v\} \neq \emptyset\}$ or, in other words, $\mR(e)$ is the set of 
edges that contain $e$ and all edges from lower levels that 
are part of the matchings in the levels but are removed
due to $e$ being added to $\hat{M}$. The weight of edge $e$ is given by $w(e)$.
As in~\cite{GP13}, we overload notation and denote the sum
of the weights of all edges in a set $S$ to be $w(S)$. 

We keep several copies of a data structure that partitions
the edges into levels while omitting different sets of edges
in each copy. For each copy, we maintain \emph{buckets}
consisting of edges and each \emph{level} consists of 
a set of buckets. An edge $e$ is in bucket $b$ if
$w(e) \in [\eps^{-b}, \eps^{-(b + 1)})$. Then, each level consists of $C -1$ continuous buckets where 
$C = \ceil{\eps^{-1}}$. We maintain $C$ copies of our
graph. In the $c$-th copy where $c \in [C]$, we remove
the edges in all buckets $i$ where $i \mod C = c$. 
Then, each level $\ell$ in copy $c$ contains buckets 
in the range $b \in [\ell \cdot C + c + 1, \dots, (\ell + 1) \cdot C + c - 1]$ which means that the ratio
the maximum weight edge and the minimum weight 
edge is any level is bounded by $\frac{\eps^{-((\ell+1)\cdot C + c)}}{\eps^{-(\ell \cdot C + c + 1)}} = \eps^{-(C - 1)} = \eps^{-O\left(\eps^{-1}\right)}$. Let $\hat{M}^c$ be the 
approximate matching computed for copy $c$. Then,
we denote copy $c$'s structures for $\hatM_{\ell}$, $\mM_{\ell}$, and $\mR(e)$ by $\hatM^c_{\ell}$ and $\mM^c_{\ell}$, and $\mR^c(e)$, respectively. 

We first prove the following lemma about the total weight
of all edges in $\mR^c(e)$ compared to the weight of $e$.

\begin{lemma}[Lemma 4.7 of~\cite{GP13}]\label{lem:bounding-approx-matching}
For any edge in $\hat{M}^c$, it holds that
\begin{align*}
    w\left(\mR^c(e)\right) \leq (1+3\eps)w(e),
\end{align*}
when $\eps < 1/2$.
\end{lemma}

\begin{proof}
    Let $\ell$ be the level that $e$ is on. Then, each
    level $\ell' < \ell$ contains at most two edges that
    are incident to an endpoint of $e$. The maximum weight of
    any edge in level $\ell'$ is $\eps^{-((\ell' + 1)\cdot C
    + c)}$. Furthermore, edge $e$ has at least 
    $\eps^{-(\ell \cdot C + c + 1)}$ weight. Thus, we 
    can upper bound $w(\mR^c(e))$ by

    \begin{align*}
        w(\mR^c(e)) &\leq w(e) + \sum_{\ell' < \ell} 2\eps^{-((\ell' + 1)\cdot C + c)}\\
        &\leq w(e) + \sum_{\ell' < \ell} 2\eps^{-((\ell' - \ell + 1)\cdot C - 1)} \cdot \eps^{-(\ell \cdot C + c + 1)}\\
        &= w(e) + \sum_{\ell' < \ell} 2\eps^{-((\ell' - \ell + 1)\cdot C - 1)} \cdot w(e)\\
        &\leq w(e) + \frac{2\eps \cdot w(e)}{1-\eps^C}\\
        &\leq w(e)(1+3\eps).
    \end{align*}
\end{proof}

Now, we show the relation between $\hat{M}^c$ and $\mM^c$;
in particular, we show that $\hat{M}^c$ is close to $\mM^c$
in size up to a small multiplicative factor.

\begin{lemma}[Lemma 4.8 of~\cite{GP13}]\label{lem:copy-orig}
    Let $\hat{M}^c$ be the approximation produced by our 
    transformation and $\mM^c$ be a maximum weighted matching in copy $c$,
    then $(1+7\eps)w(\hat{M}^c) \geq w(\mM^c)$.
\end{lemma}

\begin{proof}
    By our algorithm, each $\hatM^{c}_{\ell}$ is a 
    $(1+\eps)$-approximate weighted matching of 
    $\mM^c_{\ell}$. Then, we have:

    \begin{align*}
        w(\mM^c_{\ell}) &\leq (1+\eps)w(\hatM^c_{\ell})\\
        w(\mM^c) &\leq (1+\eps)\sum_{\ell} w(\hatM^c_{\ell}).
    \end{align*}

    Consider an edge $e = (u, v) \in \hatM^{c}_{\ell}$, then
    either: $e \in \hatM^{c}$ and $e \in \mR^c(e)$ or 
    $e \not\in M^{c}$ and $e \in \mR^c(e')$ and/or 
    $e \in \mR^c(e'')$ where $u \in e'$ and $v \in e''$
    and $e', e'' \in \hat{M}^c$.
    This means that each $e$ is mapped to at least one
    $\mR^c(e')$ for at least one edge $e' \in \hat{M}^c$.
    Then, it holds that 

    \begin{align*}
        w(R^{c}(\hat{M}^c)) &\geq \sum_{\ell} w(\hatM^c_{\ell})\\
        (1+\eps) \cdot w(R^{c}(\hat{M}^c)) &\geq (1+\eps) \cdot \sum_{\ell} w(\hatM^c_{\ell})\\
        (1+\eps) \cdot  w(R^{c}(\hat{M}^c)) &\geq w(\mM^c).
    \end{align*}

    Combining the above with~\cref{lem:bounding-approx-matching} gives

    \begin{align*}
        w(\mM^c) \leq (1+\eps) \cdot w(R^{c}(\hat{M}^c)) &\leq (1+3\eps) \cdot (1+\eps) \cdot \sum_{e \in \hatM^c} w(e) = (1+3\eps)\cdot (1+\eps) \cdot w(\hatM^c) \leq (1+7\eps) \cdot w(\hatM^c).
    \end{align*}
\end{proof}

We now show that there is at least one copy $c$ where 
$w(\mM^c) \geq (1-1/C) \cdot w(\mM)$.

\begin{lemma}[Lemma 4.9 of~\cite{GP13}]\label{lem:ratio-copy}
    There exists a copy $c$ such that $w(\mM^c) \geq (1-1/C) \cdot w(\mM)$.
\end{lemma}

\begin{proof}
    Let $\bar{M}^c$ denote the set of edges in $\mM$
    that are not present in the $c$-th copy. By our algorithm,
    each bucket is removed in exactly one copy. Then, it holds
    that 

    \begin{align*}
        \bigcup_{c} \bar{M}^c &= \mM \\
        \sum_{c} w(\bar{M}^c) &= w(\mM) 
    \end{align*}

    Since $\mM \setminus \bar{M}^c$ is a matching in 
    the $c$-th copy, we have that $w(\mM^c) \geq w(\mM) - 
    w(\bar{M}^c)$. We can sum over all $c$ copies to 
    obtain

    \begin{align*}
        \sum_{c} w(\mM^c) &\geq \sum_{c} \left(w(\mM) - w(\bar{M}^c)\right)\\
        &= C \cdot w(\mM) - \left(\sum_{c} w(\bar{M}^c)\right)\\
        &= (C - 1) \cdot w(\mM).
    \end{align*}

    This means that the average of $w(\mM^c)$ is at least
    $(1-1/C) \cdot w(\mM)$ and so there must 
    exist at least one copy $c$ where $w(\mM^c) \geq (1-1/C) \cdot w(\mM)$.
\end{proof}

Combining the above, we obtain our final theorem.

\begin{theorem}[Modified from Theorem 4.10 of~\cite{GP13}]
    For any $\eps \in (0, 1/2)$, 
    the Gupta-Peng transformation 
    produces a $(1+\eps)$-approximate MWM by running
    $O\left(\frac{\log_{(1/\eps)} (W)}{\eps}\right)$ copies of 
    a $(1+\eps')$-approximate MWM algorithm on graphs
    with maximum weight ratio $W = \eps^{-O(\eps^{-1})}$.
\end{theorem}

\begin{proof}
    There are at most $C = O(\eps^{-1})$ copies of the graph
    and in each copy that are at most 
    $O(\log_{(1/\eps)} (W))$ buckets. 
    We showed that in each level the weight ratio is upper
    bounded by $\eps^{-O(\eps^{-1})}$. Hence, we run 
    $O\left(\frac{\log_{(1/\eps)} (W)}{\eps}\right)$ 
    copies of our 
    baseline approximation algorithm on graphs with 
    weight ratios at most $\eps^{-O(\eps^{-1})}$.

    Combining~\cref{lem:copy-orig,lem:ratio-copy}, we 
    get 
    \begin{align*}
        (1+7\eps)\cdot w(\hat{M}^c) &\geq w(\mM^c) \geq 
        (1-1/C) \cdot w(\mM)\\
        \frac{1+7\eps}{1-1/C} \cdot w(\hat{M}^c) &\geq 
        w(\mM)\\
        \frac{1+7\eps}{1-\eps} \cdot w(\hat{M}^c) &\geq 
        w(\mM)\\
        (1+16\eps) \cdot w(\hat{M}^c) &\geq w(\mM).
    \end{align*}

    The final inequality holds since $\frac{1+7\eps}{1-\eps} \leq 1+16\eps$ when $\eps \in [0, 1/2]$.
    For any $(1+\eps')$-approximate MWM for $\eps' \in (0, 1/2)$, we can set $\eps$ to be appropriately small to 
    obtain that approximation.
\end{proof}

\subsection{Extensions of Gupta-Peng Transformation to Other
Models}

For the distributed and streaming settings, we use the 
transformations of Bernstein \etal~\cite{BDL21} and restate the key theorems in their paper. For the shared-memory
parallel and massively parallel computation settings, we give short proofs of how to adapt their
transformation for our settings. 

Let $W$ again be the maximum ratio between the largest weight edge and the smallest weight edge in the input graph. 
For the below theorems, whenever we write $\log_{(1/\eps)} (W)$, we assume the base of the logarithm is $1/\eps$.
The following two (modified) transformations are inspired by Bernstein \etal~\cite{BDL21}. For completeness, we present
the proofs of these transformations using our description
of the Gupta-Peng transformation above.

\begin{theorem}\label{thm:streaming-transformation}
    Given a $P(n, m, W, \eps)$-pass semi-streaming algorithm $\alg$ that computes an $(1+\eps)$-approximate maximum weight
    matching in a graph with maximum edge weight ratio $W$ and uses space $S(n, m, W, \eps)$, then 
    there exists either:
    \begin{enumerate}
        \item a $\frac{P(n, m, f(\eps), \eps)}{\eps}$-pass semi-streaming algorithm $\alg'$ that computes an $(1+16\eps)$-approximate
        maximum weight matching algorithm using $O(S(n, m, f(\eps), \eps) \cdot \log_{(1/\eps)} W)$ space,
        \item or a $P(n, m, f(\eps), \eps)$-pass semi-streaming algorithm $\alg'$ that computes an $(1+16\eps)$-approximate
        maximum weight matching algorithm using $O\left(\frac{S(n, m, f(\eps), \eps) \cdot \log_{(1/\eps)} (W)}{\eps}\right)$ space,
    \end{enumerate}
    where $f(\eps)$ is some function of $\eps$ and is
    independent of $n$ and $m$.
\end{theorem}

\begin{proof}
    When an edge passes in the stream, we calculate which
    buckets the edge belongs to in each of the copies and
    run $\alg$ on the edge assuming the edge is a new 
    edge in the stream for the corresponding bucket and 
    copy. Since we have $O\left(1/\eps\right)$ copies
    and $O(\log_{(1/\eps)} (W))$ buckets in each copy, we either:
    
    \begin{enumerate}
        \item increase the space bound by a factor of 
        $O\left(\frac{\log_{(1/\eps)} (W)}{\eps}\right)$ because 
        we keep each bucket and each copy in memory as 
        a separate instance of $\alg$, or
        \item increase the number of passes of our algorithm
        by a $O(1/\eps)$ factor where in each of the sets of 
        $P(n, m, f(\eps), \eps)$ passes, we compute the solution for
        each of the $O(\log_{(1/\eps)} (W))$ buckets as a separate 
        instance of $\alg$, increasing the space bound
        by a factor of $O(\log_{(1/\eps)} (W))$.
    \end{enumerate}

    Once we have all of the solutions for each of the buckets,
    computing the final approximate maximum matching
    can be done in memory (without using additional passes).
\end{proof}

\begin{theorem}\label{thm:distributed-transformation}
If there exists an $A(n, m, W, \eps)$-round blackboard distributed broadcast protocol $\alg$ that computes an $(1+\eps)$-approximate
maximum weight matching in a graph with maximum edge weight ratio $W$ with communication complexity 
$C(n, m, W, \eps)$ in bits, then there exists either:

\begin{enumerate}
    \item a $\frac{A(n, m, f(\eps), \eps)}{\eps}$-round 
    distributed broadcast protocol $\alg'$ that computes a $(1+16\eps)$-approximate
    maximum weight matching using $O(C(n, m, f(\eps), \eps) \cdot \log_{(1/\eps)}(W))$ bits of communication,
    \item or a $\left(A(n, m, f(\eps), \eps)\right)$-round distributed broadcast protocol $\alg'$ that computes a $(1+16\eps)$-approximate
    maximum weight matching using $O\left(\frac{C(n, m, f(\eps), \eps) \cdot \left(\log_{(1/\eps)} (W)\right)}{\eps}\right)$
    bits of communication,
\end{enumerate}
where $f(\eps)$ is some function of $\eps$.
\end{theorem} 

\begin{proof}
    Each endpoint of an edge calculates which
    buckets their adjacent edges belong to in each of the copies and
    run $\alg$ on each adjacent edge assuming the edge is part of the induced
    subgraph for the corresponding bucket and copy. 
    Since we have $O\left(1/\eps\right)$ copies
    and $O(\log_{(1/\eps)} (W))$ buckets in each copy, we either:
    
    \begin{enumerate}
        \item increase the communication complexity by a factor of 
        $O\left(\frac{\log_{(1/\eps)} (W)}{\eps}\right)$ because 
        we compute the maximum matching in each
        separate instance of $\alg$ simultaneously, or
        \item increase the number of rounds of our algorithm
        by a $O(1/\eps)$ factor where in each of the sets of 
        $P(n, m, f(\eps), \eps)$ rounds, we compute the solution for
        each of the $O(\log_{(1/\eps)} (W))$ buckets as a separate 
        instance of $\alg$ and then proceed with the next copy, 
        increasing the communication complexity        
        by a factor of $O(\log_{(1/\eps)} (W))$.
    \end{enumerate}

    Once we have all of the solutions for each of the buckets written on the blackboard, 
    then we can compute $\hat{M}^c$ for each copy and return the maximum among the copies.
\end{proof}

We give the following transformations for the shared-memory parallel and massively parallel computation models.

\begin{theorem}\label{thm:parallel}
If there exists a parallel algorithm $\alg$ that computes an $(1+\eps)$-approximate maximum weight matching 
in a graph with maximum edge weight ratio $W$ with $B(n, m, W, \eps)$ work and $D(n, m, W, \eps)$ depth, then
there exists a $O\left(\frac{W(n, m, f(\eps), \eps)\cdot \log_{(1/\eps)} (W)}{\eps}\right)$ work 
and $O\left(D(n, m, f(\eps), \eps) \cdot \log_{(1/\eps)}(W)\right)$ 
depth parallel
algorithm $\alg'$ that gives a $(1+16\eps)$-approximate maximum weight matching, 
where $f(\eps)$ is some function of $\eps$.
\end{theorem}

\begin{proof}
To obtain $\alg'$, we maintain each of the $C = O\left(\frac{1}{\eps}\right)$ subgraphs $\{G_1, \dots, 
G_C\}$ of the Gupta-Peng transformation in parallel incurring a factor of $O\left(\frac{\log_{(1/\eps)} (W)}{\eps}\right)$ 
additional total work. The depth is now $O\left(D(n, m, f(\eps), \eps) \cdot \log_{(1/\eps)}(W)\right)$ 
since we now need to compute the matching per 
level sequentially. The 
computation for each subgraph and within the levels in 
each subgraph can be done in parallel and the depth is a function of the maximum ratio of 
weights in the graph in each level which is $f(\eps)$.
\end{proof}

\begin{theorem}\label{thm:mpc}
    If there exists a MPC algorithm $\alg$ that computes an $(1+\eps)$-approximate maximum weight matching 
    in a graph with maximum edge weight ratio $W$ in $R(n, m, W, \eps)$ rounds, $S(n, m, W, \eps)$ space per machine
    and $T(n, m, W, \eps)$ total space, then
    there exists a $O\left(R(n, m, f(\eps), \eps)\right)$ rounds, 
    $O(S(n, m, f(\eps), \eps) + n \cdot \log_{(1/\eps)}(W))$ space per machine, and 
    $O\left(\frac{T(n, m, f(\eps), \eps) \cdot \log_{(1/\eps)} (W)}{\eps}\right)$ total space MPC
    algorithm $\alg'$ that gives a $(1+16\eps)$-approximate maximum weight matching, 
    where $f(\eps)$ is some function of $\eps$.
\end{theorem}

\begin{proof}
    To obtain $\alg'$, we maintain each of the $C = O\left(\frac{1}{\eps}\right)$ subgraphs $\{G_1, \dots, 
    G_C\}$ of the Gupta-Peng transformation in parallel with each level 
    partitioned across machines in the same way as the original algorithm.
    The number of rounds is equal to the number of rounds for any particular instance so it is 
    equal to $O\left(R(n, m, f(\eps), \eps)\right)$ since each instance has maximum weight ratio $f(\eps)$. 
    Since each instance can be handled in parallel by the algorithm, the space per instance is $O(S(n, m, f(\eps), \eps))$. Once the matching
    per level is computed, all of the levels for the same copy are put
    onto one matching. Because each level is a matching and since there
    are $O\left(\log_{(1/\eps)}(W)\right)$ levels. The total space per
    machine that is used is $O\left(n \cdot \log_{(1/\eps)}(W)\right)$.
    The total space is now $O\left(\frac{T(n, m, f(\eps), \eps) \cdot \log_{(1/\eps)}(W)}{\eps}\right)$ since the 
    computation for each subgraph and within the levels in 
    each subgraph can be done in parallel and each requires $T(n, m, f(\eps), \eps)$ total space.
\end{proof}

\bibliographystyle{alpha}
\bibliography{ref}

\newcommand{\etalchar}[1]{$^{#1}$}
\begin{thebibliography}{FKM{\etalchar{+}}05}

\bibitem[AG11]{AG11}
Kook~Jin Ahn and Sudipto Guha.
\newblock Linear programming in the semi-streaming model with application to
  the maximum matching problem.
\newblock In {\em Proceedings of the 38th International Conference on Automata,
  Languages and Programming - Volume Part II}, ICALP'11, page 526–538,
  Berlin, Heidelberg, 2011. Springer-Verlag.

\bibitem[AG18]{AG18}
Kook~Jin Ahn and Sudipto Guha.
\newblock Access to data and number of iterations: Dual primal algorithms for
  maximum matching under resource constraints.
\newblock {\em ACM Trans. Parallel Comput.}, 4(4), Jan 2018.

\bibitem[AJJ{\etalchar{+}}22]{assadi2022semi}
Sepehr Assadi, Arun Jambulapati, Yujia Jin, Aaron Sidford, and Kevin Tian.
\newblock Semi-streaming bipartite matching in fewer passes and optimal space.
\newblock In {\em Proceedings of the 2022 Annual ACM-SIAM Symposium on Discrete
  Algorithms (SODA)}, pages 627--669. SIAM, 2022.

\bibitem[AKSY20]{AKSY20}
Sepehr Assadi, Gillat Kol, Raghuvansh~R. Saxena, and Huacheng Yu.
\newblock Multi-pass graph streaming lower bounds for cycle counting, max-cut,
  matching size, and other problems.
\newblock In {\em 2020 IEEE 61st Annual Symposium on Foundations of Computer
  Science (FOCS)}, pages 354--364, 2020.

\bibitem[ALT21]{ALT21}
Sepehr Assadi, S.~Cliff Liu, and Robert~E. Tarjan.
\newblock {\em An Auction Algorithm for Bipartite Matching in Streaming and
  Massively Parallel Computation Models}, pages 165--171.
\newblock 2021.

\bibitem[Ass22]{Assadi22}
Sepehr Assadi.
\newblock A two-pass (conditional) lower bound for semi-streaming maximum
  matching.
\newblock In Joseph~(Seffi) Naor and Niv Buchbinder, editors, {\em Proceedings
  of the 2022 {ACM-SIAM} Symposium on Discrete Algorithms, {SODA} 2022, Virtual
  Conference / Alexandria, VA, USA, January 9 - 12, 2022}, pages 708--742.
  {SIAM}, 2022.

\bibitem[BDL21]{BDL21}
Aaron Bernstein, Aditi Dudeja, and Zachary Langley.
\newblock A framework for dynamic matching in weighted graphs.
\newblock In {\em Proceedings of the 53rd Annual ACM SIGACT Symposium on Theory
  of Computing}, STOC 2021, page 668–681, New York, NY, USA, 2021.
  Association for Computing Machinery.

\bibitem[Ber81]{bertsekas1981new}
Dimitri~P Bertsekas.
\newblock A new algorithm for the assignment problem.
\newblock {\em Mathematical Programming}, 21(1):152--171, 1981.

\bibitem[BFS12]{BFS12}
Guy~E. Blelloch, Jeremy~T. Fineman, and Julian Shun.
\newblock Greedy sequential maximal independent set and matching are parallel
  on average.
\newblock In {\em {ACM} Symposium on Parallelism in Algorithms and
  Architectures (SPAA)}, pages 308--317, 2012.

\bibitem[BHH19]{BHH19}
Soheil Behnezhad, MohammadTaghi Hajiaghayi, and David~G. Harris.
\newblock Exponentially faster massively parallel maximal matching.
\newblock In David Zuckerman, editor, {\em 60th {IEEE} Annual Symposium on
  Foundations of Computer Science, {FOCS} 2019, Baltimore, Maryland, USA,
  November 9-12, 2019}, pages 1637--1649. {IEEE} Computer Society, 2019.

\bibitem[BKS17]{beame2017communication}
Paul Beame, Paraschos Koutris, and Dan Suciu.
\newblock Communication steps for parallel query processing.
\newblock {\em Journal of the ACM (JACM)}, 64(6):1--58, 2017.

\bibitem[BOS{\etalchar{+}}13]{Birn2013}
Marcel Birn, Vitaly Osipov, Peter Sanders, Christian Schulz, and Nodari
  Sitchinava.
\newblock Efficient parallel and external matching.
\newblock In {\em International Conference on Parallel Processing (Euro-Par)},
  page 659–670, 2013.

\bibitem[DGS86]{demange1986multi}
Gabrielle Demange, David Gale, and Marilda Sotomayor.
\newblock Multi-item auctions.
\newblock {\em Journal of political economy}, 94(4):863--872, 1986.

\bibitem[DNO14]{DNO14}
Shahar Dobzinski, Noam Nisan, and Sigal Oren.
\newblock Economic efficiency requires interaction.
\newblock In {\em Proceedings of the Forty-Sixth Annual ACM Symposium on Theory
  of Computing}, STOC '14, page 233–242, New York, NY, USA, 2014. Association
  for Computing Machinery.

\bibitem[Edm65a]{edmonds1965maximum}
Jack Edmonds.
\newblock Maximum matching and a polyhedron with 0, 1-vertices.
\newblock {\em Journal of Research of the National Bureau of Standards B},
  69:125--130, 1965.

\bibitem[Edm65b]{Edmonds1965PathsTA}
Jack Edmonds.
\newblock Paths, trees, and flowers.
\newblock {\em Canadian Journal of Mathematics}, 17:449 -- 467, 1965.

\bibitem[FKM{\etalchar{+}}05]{feigenbaum2005graph}
Joan Feigenbaum, Sampath Kannan, Andrew McGregor, Siddharth Suri, and Jian
  Zhang.
\newblock On graph problems in a semi-streaming model.
\newblock {\em Theoretical Computer Science}, 348(2-3):207--216, 2005.

\bibitem[FMU22]{FMU22}
Manuela Fischer, Slobodan Mitrovi\'{c}, and Jara Uitto.
\newblock Deterministic $(1+\eps)$-approximate maximum matching with
  poly($1/\eps$) passes in the semi-streaming model and beyond.
\newblock In {\em Proceedings of the 54th Annual ACM SIGACT Symposium on Theory
  of Computing}, STOC 2022, page 248–260, New York, NY, USA, 2022.
  Association for Computing Machinery.

\bibitem[FN18]{FischerN18}
Manuela Fischer and Andreas Noever.
\newblock Tight analysis of parallel randomized greedy {MIS}.
\newblock In {\em {ACM-SIAM} Symposium on Discrete Algorithms}, pages
  2152--2160, 2018.

\bibitem[GGM22]{GGM22}
Mohsen Ghaffari, Christoph Grunau, and Slobodan Mitrovi\'c.
\newblock Massively parallel algorithms for $b$-matching.
\newblock In {\em ACM Symposium on Parallelism in Algorithms and Architectures
  (SPAA)}, 2022.

\bibitem[GKK]{GKK12}
Ashish Goel, Michael Kapralov, and Sanjeev Khanna.
\newblock {\em On the communication and streaming complexity of maximum
  bipartite matching}, pages 468--485.

\bibitem[GKMS19]{GKMS19}
Buddhima Gamlath, Sagar Kale, Slobodan Mitrovic, and Ola Svensson.
\newblock Weighted matchings via unweighted augmentations.
\newblock In {\em Proceedings of the 2019 ACM Symposium on Principles of
  Distributed Computing}, PODC '19, page 491–500, New York, NY, USA, 2019.
  Association for Computing Machinery.

\bibitem[GP13]{GP13}
Manoj Gupta and Richard Peng.
\newblock Fully dynamic {(1+} e)-approximate matchings.
\newblock In {\em {FOCS}}, pages 548--557. {IEEE} Computer Society, 2013.

\bibitem[GSZ11]{goodrich2011sorting}
Michael~T Goodrich, Nodari Sitchinava, and Qin Zhang.
\newblock Sorting, searching, and simulation in the mapreduce framework.
\newblock In {\em International Symposium on Algorithms and Computation}, pages
  374--383. Springer, 2011.

\bibitem[Har06]{Harvey2006AlgebraicSA}
Nicholas J.~A. Harvey.
\newblock Algebraic structures and algorithms for matching and matroid
  problems.
\newblock {\em 2006 47th Annual IEEE Symposium on Foundations of Computer
  Science (FOCS'06)}, pages 531--542, 2006.

\bibitem[HK71]{HK71}
John~E. Hopcroft and Richard~M. Karp.
\newblock A n5/2 algorithm for maximum matchings in bipartite.
\newblock In {\em 12th Annual Symposium on Switching and Automata Theory (swat
  1971)}, pages 122--125, 1971.

\bibitem[HS22]{HS22}
Shang{-}En Huang and Hsin{-}Hao Su.
\newblock (1-{\(\epsilon\)})-approximate maximum weighted matching in
  poly(1/{\(\epsilon\)}, log n) time in the distributed and parallel settings.
\newblock {\em CoRR}, abs/2212.14425, 2022.

\bibitem[J{\'a}J92]{jaja1992introduction}
Joseph J{\'a}J{\'a}.
\newblock {\em An introduction to parallel algorithms}.
\newblock Addison Wesley Longman Publishing Co., Inc., 1992.

\bibitem[JLS19]{JLS19}
Arun Jambulapati, Yang~P. Liu, and Aaron Sidford.
\newblock Parallel reachability in almost linear work and square root depth.
\newblock In {\em 2019 IEEE 60th Annual Symposium on Foundations of Computer
  Science (FOCS)}, pages 1664--1686, 2019.

\bibitem[Kap13]{kapralov2013better}
Michael Kapralov.
\newblock Better bounds for matchings in the streaming model.
\newblock In {\em Proceedings of the twenty-fourth annual ACM-SIAM symposium on
  Discrete algorithms}, pages 1679--1697. SIAM, 2013.

\bibitem[KRZ21]{KRZ21}
Christian Konrad, Peter Robinson, and Viktor Zamaraev.
\newblock Robust lower bounds for graph problems in the blackboard model of
  communication.
\newblock {\em CoRR}, abs/2103.07027, 2021.

\bibitem[KSV10]{karloff2010model}
Howard Karloff, Siddharth Suri, and Sergei Vassilvitskii.
\newblock A model of computation for mapreduce.
\newblock In {\em Proceedings of the twenty-first annual ACM-SIAM symposium on
  Discrete Algorithms}, pages 938--948. SIAM, 2010.

\bibitem[Kö16]{Konig1916}
D.~König.
\newblock Über graphen und ihre anwendung auf determinantentheorie und
  mengenlehre.
\newblock {\em Mathematische Annalen}, 77:453--465, 1916.

\bibitem[LPSP15]{LPPS15}
Zvi Lotker, Boaz Patt-Shamir, and Seth Pettie.
\newblock Improved distributed approximate matching.
\newblock {\em J. ACM}, 62(5), nov 2015.

\bibitem[LS20]{LS20}
Yang~P. Liu and Aaron Sidford.
\newblock {\em Faster Energy Maximization for Faster Maximum Flow}, page
  803–814.
\newblock Association for Computing Machinery, New York, NY, USA, 2020.

\bibitem[LSZ20]{liu2020breaking}
S~Cliff Liu, Zhao Song, and Hengjie Zhang.
\newblock Breaking the $ n $-pass barrier: A streaming algorithm for maximum
  weight bipartite matching.
\newblock {\em arXiv preprint arXiv:2009.06106}, 2020.

\bibitem[Mad13]{madry13}
Aleksander Madry.
\newblock Navigating central path with electrical flows: From flows to
  matchings, and back.
\newblock In {\em 2013 IEEE 54th Annual Symposium on Foundations of Computer
  Science}, pages 253--262, 2013.

\bibitem[MS04]{MS04}
M.~Mucha and P.~Sankowski.
\newblock Maximum matchings via gaussian elimination.
\newblock In {\em 45th Annual IEEE Symposium on Foundations of Computer
  Science}, pages 248--255, 2004.

\bibitem[MV80]{MV80}
Silvio Micali and Vijay~V. Vazirani.
\newblock An $o(\sqrt{|v|} \cdot |e|)$ algorithm for finding maximum matching
  in general graphs.
\newblock In {\em 21st Annual Symposium on Foundations of Computer Science
  (sfcs 1980)}, pages 17--27, 1980.

\bibitem[Nis21]{nisan2021demand}
Noam Nisan.
\newblock The demand query model for bipartite matching.
\newblock In {\em Proceedings of the 2021 ACM-SIAM Symposium on Discrete
  Algorithms (SODA)}, pages 592--599. SIAM, 2021.

\bibitem[R{\etalchar{+}}90]{ramachandran1990parallel}
Vijaya RAMACHANDRAN et~al.
\newblock Parallel algorithms for shared-memory machines.
\newblock In {\em Algorithms and Complexity}, pages 869--941. Elsevier, 1990.

\bibitem[SV82]{shiloach1982n2log}
Yossi Shiloach and Uzi Vishkin.
\newblock An $o(n^2\log n)$ parallel max-flow algorithm.
\newblock {\em Journal of Algorithms}, 3(2):128--146, 1982.

\bibitem[SW17]{SW17}
Daniel Stubbs and Virginia~Vassilevska Williams.
\newblock Metatheorems for dynamic weighted matching.
\newblock In Christos~H. Papadimitriou, editor, {\em 8th Innovations in
  Theoretical Computer Science Conference, {ITCS} 2017, January 9-11, 2017,
  Berkeley, CA, {USA}}, volume~67 of {\em LIPIcs}, pages 58:1--58:14. Schloss
  Dagstuhl - Leibniz-Zentrum f{\"{u}}r Informatik, 2017.

\bibitem[ZH23]{ZH23}
Da~Wei Zheng and Monika Henzinger.
\newblock Multiplicative auction algorithm for approximate maximum weight
  bipartite matching.
\newblock {\em arXiv preprint arXiv:2301.09217}, 2023.

\end{thebibliography}

\end{document}